\documentclass[11pt,onecolumn,draftcls]{IEEEtran}

\pdfoutput=0

\usepackage{amsmath,amssymb,amsfonts,amsthm}
\usepackage{ucs}
\usepackage{color}
\usepackage[utf8x]{inputenc}
\usepackage[dvips]{graphics}
\usepackage{psfrag}
\usepackage{multirow}
\usepackage{rotating}
\usepackage{url}

\def\diag{\mathop{\rm diag}}

\DeclareSymbolFont{bbold}{U}{bbold}{m}{n}
\DeclareSymbolFontAlphabet{\mathbbold}{bbold}

\newtheorem{thm}{Theorem}
\newtheorem{lem}{Lemma}

\makeatletter
\newif\if@borderstar
\def\bordermatrix{\@ifnextchar*{%
\@borderstartrue\@bordermatrix@i}{\@borderstarfalse\@bordermatrix@i*}%
}
\def\@bordermatrix@i*{\@ifnextchar[{\@bordermatrix@ii}{\@bordermatrix@ii[()]}}
\def\@bordermatrix@ii[#1]#2{%
\begingroup
\m@th\@tempdima8.75\p@\setbox\z@\vbox{%
\def\cr{\crcr\noalign{\kern 2\p@\global\let\cr\endline }}%
\ialign {$##$\hfil\kern 2\p@\kern\@tempdima & \thinspace %
\hfil $##$\hfil && \quad\hfil $##$\hfil\crcr\omit\strut %
\hfil\crcr\noalign{\kern -\baselineskip}#2\crcr\omit %
\strut\cr}}%
\setbox\tw@\vbox{\unvcopy\z@\global\setbox\@ne\lastbox}%
\setbox\tw@\hbox{\unhbox\@ne\unskip\global\setbox\@ne\lastbox}%
\setbox\tw@\hbox{%
$\kern\wd\@ne\kern -\@tempdima\left\@firstoftwo#1%
\if@borderstar\kern2pt\else\kern -\wd\@ne\fi%
\global\setbox\@ne\vbox{\box\@ne\if@borderstar\else\kern 2\p@\fi}%
\vcenter{\if@borderstar\else\kern -\ht\@ne\fi%
\unvbox\z@\kern-\if@borderstar2\fi\baselineskip}%
\if@borderstar\kern-2\@tempdima\kern2\p@\else\,\fi\right\@secondoftwo#1 $%
}\null \;\vbox{\kern\ht\@ne\box\tw@}%
\endgroup
}
\makeatother

\title{Capacity of DNA Data Embedding Under Substitution
  Mutations}

\author{F{\'e}lix Balado\thanks{F.~Balado is with the School of
    Computer Science and Informatics, University College Dublin,
    Belfield, Dublin 4, Ireland. E-mail: felix@ucd.ie. Preliminary
    versions of this work were presented at the SPIE Media Forensics and
    Security XII conference (January 2010) and at the IEEE ICASSP
    conference (March 2010).}}

\begin{document}
\maketitle

\begin{abstract} 
  A number of methods have been proposed over the last decade for
  encoding information using deoxyribonucleic acid (DNA), giving rise
  to the emerging area of DNA data embedding. Since a DNA sequence is
  conceptually equivalent to a sequence of quaternary symbols (bases),
  DNA data embedding (diversely called DNA watermarking or DNA
  steganography) can be seen as a digital communications problem where
  channel errors are tantamount to mutations of DNA bases.  Depending
  on the use of coding or noncoding DNA hosts, which, respectively,
  denote DNA segments that can or cannot be translated into proteins,
  DNA data embedding is essentially a problem of communications with
  or without side information at the encoder. In this paper the
  Shannon capacity of DNA data embedding is obtained for the case in
  which DNA sequences are subject to substitution mutations modelled
  using the Kimura model from molecular evolution studies.  Inferences
  are also drawn with respect to the biological implications of some
  of the results presented.
  \end{abstract}

\IEEEpeerreviewmaketitle

\section{Introduction}
\label{sec:introduction}
\IEEEPARstart{T}{he} last ten years have witnessed the proposal of
numerous practical
methods~\cite{clelland99a,cox01:long_term,Shimanovsky02,wong03:organ_data,arita04:secret,modegi05:water_embed,yachie07:alignnment,heider07:dna_based_watermarking,heider09:dna_noncoding}
for encoding nongenetic information using DNA molecules as a medium
both \textit{in vitro} and \textit{in vivo}.  A conspicuous use of
these techniques recently took place when Craig Venter's group
produced the first artificial bacteria including ``watermarked''
information~\cite{daniel08:complete}. All of these information
encoding proposals hinge on the fact that DNA molecules ---which
encode genetic information in all living organisms, except for some
viruses--- are conceptually equivalent to sequences of quaternary
symbols. Therefore DNA data embedding is in essence an instance of
digital communications in which channel errors are tantamount to
mutations of DNA components. The two broad fields of application of
DNA data embedding techniques are: 1) the use of DNA strands as
self-replicating nano-memories able to store huge amounts of data in
an ultra-compact way; and 2) security and tracking applications for
genetic material afforded by embedding nongenetic information in DNA
(DNA watermarking, steganography, and fingerprinting).

The most basic information theoretical issue in DNA data embedding is
the establishment of the upper limit on the amount of information that
can be reliably embedded within DNA under a given level of mutations,
that is, its Shannon capacity~\cite{shannon48:math}. In this paper we
obtain the capacity of DNA data embedding under substitution mutations
---which randomly switch the value of bases in a DNA sequence---
modelled through a symmetric memoryless channel which was firstly used
to study molecular evolution by Kimura~\cite{kimura80}.  The capacity
problem can be straightforwardly tackled when no side information is
used by the encoder. The side-informed scenario requires more
attention for reasons that will become clear later, and thus occupies
us for the best part of this paper. Some biological implications at
large of these information theoretical results are also discussed. In
particular, the non side-informed scenario happens to be closely
connected to previous studies by May, Battail, and other authors that
have tried to apply information theoretical concepts to molecular
biology.

\begin{table*}[t]
  \centering
  {\scriptsize
    \setlength{\tabcolsep}{1pt}
    
    \begin{tabular}{c | c c c c c c c c c c c c c c c c c c c c c}
      $x'$&Ala&Arg&Asn&Asp&Cys&Gln&Glu&Gly&His&Ile&Leu&Lys&Met&Phe&Pro&Ser&Thr&Trp&Tyr&Val&\textit{Stp}\\
      \hline
      \multirow{6}{5mm}{\hspace{0.6cm}
        $\mathcal{S}_{x'}$}
          &GCA&AGA&AAC&GAC&TGC&CAA&GAA&GGA&CAC&ATA&CTA&AAA&\underbar{ATG}&TTC&CCA&AGC&ACA&TGG&TAC&GTA&TAA\\
          &GCC&AGG&AAT&GAT&TGT&CAG&GAG&GGC&CAT&ATC&CTC&AAG&   &TTT&CCC&AGT&ACC&   &TAT&GTC&TAG\\
          &GCT&CGA&   &   &   &   &   &GGT&   &ATT&CTT&   &   &   &CCT&TCA&ACT&   &   &GTT&TGA\\
          &GCG&CGC&   &   &   &   &   &GGG&   &   &\underbar{CTG}&   &   &   &CCG&TCC&ACG&   &   &GTG&   \\
          &   &CGT&   &   &   &   &   &   &   &   &TTA&   &   &   &   &TCT&   &   &   &   &   \\
          &   &CGG&   &   &   &   &   &   &   &   &\underbar{TTG}&   &   &   &   &TCG&   &   &   &   &   \\
          \hline
        $|\mathcal{S}_{x'}|$ 
          &4  &6  &2  &2  &2  &2  &2  &4  &2  &3  &6  &2  &1  &2  &4  &6  &4  &1  &2  &4  &3  \\
    \end{tabular}
}
\caption{Equivalences between amino acids and codons (genetic
  code). Start codons, which double as regular codons, are underlined.}
  \label{tab:amino}
\end{table*}

\section{Preliminary Concepts and Assumptions}
\label{sec:preliminary-concepts} 
Chemically, DNA is formed by two backbone strands helicoidally twisted
around each other, and mutually attached by means of two \textit{base}
sequences. The four possible bases are the molecules adenine,
cytosine, thymine, and guanine, abbreviated A, C, T and G,
respectively. Only the pairings A-T and C-G can exist between the two
strands, which is why each of the two base sequences is completely
determined by the other, and also why the length of a DNA molecule is
measured in base pairs (bp).  According to this brief description, the
interpretation of DNA as a one-dimensional discrete digital signal is
straightforward: any of the two strands constitutes a digital sequence
formed by symbols from a quaternary alphabet.

As regards the biological meaning of DNA, for the purposes of our
analysis it suffices to know that \textit{codons} ---the minimal
biological ``codewords''--- are formed by triplets of consecutive
bases in a base sequence. Given any three consecutive bases there is
no ambiguity in the codon they stand for, since there is only one
direction in which a base sequence can be read. In molecular biology
this is called the 5'--3' direction, in reference to certain chemical
feature points in a 
DNA backbone strand. The two strands in a DNA molecule are read in
opposite directions, and because of this and of their complementarity
they are termed antiparallel. Groups of consecutive codons in some
special regions of a DNA sequence can be translated into a series of
chemical compounds called \textit{amino acids} via transcription to
the intermediary ribonucleic acid (RNA) molecule.  RNA is similar to
DNA but single stranded and with uracil (abbreviated U) replacing
thymine.  Amino acids are sequentially assembled in the same order
imposed by the codon sequence. The result of this assembling process
are proteins, which are the basic compounds of the chemistry of
life. There are $4^3=64$ possible codons, since they are triplets of
$4$-ary symbols. Crucially, there are only $20$ possible amino acids,
mapped to the $64$ codons according to the so-called \textit{genetic
  code} in Table~\ref{tab:amino}, which will be explained in more
detail later. The genetic code effectively implements built-in
redundancy in terms of protecting protein expression.

The genome of an organism is the ensemble of all its DNA.  Segments of
a genome that can be translated into proteins through the process
described above are called \textit{coding} DNA (cDNA), whereas those
segments that never get translated are called \textit{noncoding} DNA
(ncDNA). A \textit{gene} is a cDNA segment, or group of segments,
which encodes one single protein, and which is flanked by certain
start and stop codons (see Table~\ref{tab:amino}) plus other
markers. 

Finally, for each base sequence
there are three different reading frames which determine three
different codon sequences. The correct reading frame is marked by the
position of a start codon. 

The main assumptions that we will make in our analysis are the following ones:
\begin{itemize}
\item \textit{ncDNA can be freely appended or overwritten}. Although
  ncDNA does not encode genes, this assumption does not always
  hold true. This is because certain ncDNA regions act as promoters
  for gene expression, or are transcribed into regulatory RNA (but not
  translated into proteins). However this working hypothesis is valid
  in suitably chosen ncDNA regions, as proved by several
  researchers~\cite{wong03:organ_data,yachie07:alignnment} employing
  live organisms.
\item \textit{cDNA can be freely modified as long as the genetic code
    is observed}. This is the classic standard assumption supporting
  the validity of the genetic code. In practice living organisms
  feature preferred codon statistics, which, if modified, might alter
  gene expression (for instance translation times, among other
  effects). Therefore we will also discuss codon statistics
  preservation in our analysis. Finally we must mention that we will
  only consider nonoverlapping genes (either on the same or on
  opposite strands). However overlapping genes are in any case rare
  occurrences, except in very compact genomes.
\end{itemize}

\textbf{Notation.} Calligraphic letters ($\mathcal{X}$) denote sets;
$|\mathcal{X}|$ is the cardinality of $\mathcal{X}$. Boldface letters
($\mathbf{x}$) denote row vectors, and $\mathbf{1}$ is an all-ones
vector. If a Roman letter is used both in uppercase ($X$) and
lowercase ($x$), the two forms denote a random variable and a
realisation of it, respectively. $p(X=x)$ is the probability mass
function (pmf) of $X$; we will simply write $p(x)$ when the variable
is clear from the context. $E[X]$ is the mathematical expectation of
$X$, and $H(X)$ its entropy. Also, $h(q)$ is the entropy of a
Bernoulli($q$) random variable.  $I(X;Y)$ is the mutual information
between $X$ and $Y$. Logarithms are base 2, unless explicitly
indicated otherwise. The Hamming distance between vectors $\mathbf{x}$
and $\mathbf{y}$ is denoted by $d_H(\mathbf{x},\mathbf{y})$.

A ncDNA sequence will be denoted by a vector
$\mathbf{x}^b=[x_1,x_2,\cdots,x_n]$, whose elements are consecutive
bases from a base sequence. That is,
$x_i\in\mathcal{X}\triangleq\{\mathrm{A,C,T,G}\}$, the $4$-ary set of
possible bases. A cDNA sequence will be denoted by a vector of vectors
$\mathbf{x}^c=[\mathbf{x}_1,\mathbf{x}_2,\cdots, \mathbf{x}_{n}]$
whose elements are consecutive codons from one of the two antiparallel
base sequences, assuming a suitable reading frame among the three
possible ones. Therefore, $\mathbf{x}_i\in \mathcal{X}^3$.  We denote
by $x'_i\triangleq\alpha(\mathbf{x}_i)\in\mathcal{X}'$ the amino acid
into which a codon $\mathbf{x}_i$ uniquely translates, which is
further discussed below. Also
$\mathbf{x}'=\alpha(\mathbf{x}^c)=[x_1',x_2',\cdots,x_n']$ denotes the
unique amino acid sequence established by~$\mathbf{x}^c$, usually
called the primary structure.  Using the standard three-letter
abbreviations of the amino acid names, we define the set
$\mathcal{X'}\triangleq\{$Ala, Arg, Asn, Asp, Cys, Gln, Glu, Gly, His,
Ile, Leu, Lys, Met, Phe, Pro, Ser, Thr, Trp, Tyr, Val,
\textit{Stp}\}. The subset of codons associated with amino acid
$x'\in\mathcal{X'}$, that is,
$\mathcal{S}_{x'}\triangleq\{\mathbf{x}\in\mathcal{X}^3|\alpha(\mathbf{x})=x'\}$,
is established by the genetic code shown in Table~\ref{tab:amino}.
The ensemble of stop codons, that is, the stop symbol \textit{Stp}, is
loosely classed as an ``amino acid'' for notational convenience,
although it does not actually map to any compound but rather indicates
the end of a gene. We call the number of codons $|\mathcal{S}_{x'}|$
mapping to amino acid~$x'$ the \textit{multiplicity} of $x'$. Due to
the uniqueness of the mapping from codons to amino acids, see that
$\mathcal{S}_{x'}\cap \mathcal{S}_{y'}=\emptyset$ for $x'\neq
y'\in\mathcal{X'}$, and that $\sum_{x'\in\mathcal{X}'}
|\mathcal{S}_{x'}|=|\mathcal{X}|^3=64$ since
$\cup_{x'\in\mathcal{X}'}\mathcal{S}_{x'}=\mathcal{X}^3$.  Finally, an
example of a cDNA sequence may be for instance
$\mathbf{x}^c=[[\mathrm{T,A,T}],[\mathrm{T,G,C}]]$, which would encode
the amino acid sequence
$\mathbf{x}'=\alpha(\mathbf{x}^c)=[\mathrm{Tyr},\mathrm{Cys}]$. The
corresponding base sequence would be
$\mathbf{x}^b=[\mathrm{T,A,T,T,G,C}]$.

\subsection{Mutation channel model} 
\label{sec:robustn-under-mutat}
As mentioned in the introduction, an information-carrying DNA molecule
undergoing mutations can be readily seen as a digital signal
undergoing a noisy communications channel, which we may term
``mutation channel'' in this context.  We will only consider herein
substitution mutations (also called point mutations), that is, those
that randomly switch letters from the DNA alphabet. We will assume
that mutations are mutually independent, which is a worst-case
scenario in terms of capacity. Therefore we are assuming that the
channel is memoryless and thus accepts a single-letter
characterisation; consequently, we will drop vector element subindices
whenever this is unambiguous and notationally convenient.

We will model the channel by means of the two-parameter Kimura model
of nucleotide substitution~\cite{kimura80}. This consists of a
$4\times 4$ transition probability matrix $\Pi= [p(Z=z|Y=y)]$, where
$z,y\in\mathcal{X}$, and which presents the following structure:

\vspace{-.5cm}
{\small
\begin{equation}
  \label{eq:pi_symm}
  \Pi\triangleq\!\!\!\bordermatrix[{[]}]{%
    &\textrm{A} & \textrm{C} & \textrm{T} & \textrm{G} \cr
    &1-q&\frac{\gamma}{3}q&\frac{\gamma}{3}q&(1-\frac{2\gamma}{3})q\cr
    &\frac{\gamma}{3}q&1-q&(1-\frac{2\gamma}{3})q&\frac{\gamma}{3}q\cr
    &\frac{\gamma}{3}q&(1-\frac{2\gamma}{3})q&1-q&\frac{\gamma}{3}q\cr
    &(1-\frac{2\gamma}{3})q&\frac{\gamma}{3}q&\frac{\gamma}{3}q&1-q}~\begin{array}{c}
    \mathrm{A}\\
    \mathrm{C}\\
    \mathrm{T}\\
    \mathrm{G}
  \end{array} 
\end{equation}}
From this definition, the probability of base substitution mutation,
or base substitution mutation rate, is 
\begin{equation}
  q=p(Z\neq y|Y=y)=\sum_{z\neq
    y}p(Z=z|Y=y),\label{eq:q}
\end{equation}
for any $y\in\mathcal{X}$, whereas it must hold that
$0\le\gamma\le 3/2$ so that row probabilities add up to one. The particular structure of~$\Pi$ aims at reflecting the fact that DNA bases belong to one of two
categories according their chemical structure: purines,
$\mathcal{R}\triangleq\{\mathrm{A},\mathrm{G}\}$, or pyrimidines,
$\mathcal{Y}\triangleq\{\mathrm{C},\mathrm{T}\}$.  There are two types
of base substitutions associated to these categories, which in
biological nomenclature are:
\begin{itemize}
\item Base \textit{transitions}: those that preserve the category which
  the base belongs to. In this case the model
  establishes that $p(Z=z|Y=y)=(1-2\gamma/3)q$ for $z\ne y$ when
  either both $z,y\in\mathcal{R}$ or both $z,y \in\mathcal{Y}$.
\item Base \textit{transversions}: those that switch the base category. In this
  case the model establishes 
  that $p(Z=z|Y=y)=(\gamma/3) q$ for $z\ne y$ when $z\in\mathcal{Y}$ and
  $z\in\mathcal{R}$, or vice versa.
\end{itemize}
The channel model~\eqref{eq:pi_symm} can incorporate any given
transition/transversion ratio $\varepsilon$ by setting
$\gamma=3/(2(\varepsilon+1))$. Estimates of~$\varepsilon$ given
in~\cite{purvis97:ti_tv} for the DNA of different organisms range
between $0.89$ and $18.67$, corresponding to $\gamma$ between $0.07$
and $0.79$. This range of $\varepsilon$ reflects the fact that base
transitions are generally much more likely than base transversions due
to the chemical similarity among compounds in the same category, that
is, $\varepsilon>1/2$ virtually always in every organism, and
therefore $\gamma<1$.  However many mutation estimation studies focus
only on the determination of~$q$ (see for
instance~\cite{kunkel04:dna_replication}), and then one may assume the
simplification $\gamma=1$ in the absence of further details. We will
make observations at several points for this particular case, which is
known as the Jukes-Cantor model in molecular evolution studies. In
this situation all off-diagonal entries of $\Pi$ are equal, that is,
$p(Z=z|Y=y)=q/3$ for all $z\neq y$.

Note that the mutation model that we have chosen implies a symmetric
channel, since all rows (columns) of $\Pi$ contain the same four
probabilities. Among the memoryless models used in molecular
evolution, the Kimura model is the one with higher number of
parameters which still yields a symmetric channel. As it is well
known, this is advantageous in capacity computations and will be
exploited whenever possible.  In the most general case a
time-reversible substitution mutations model may have up to 9
independent parameters, and yield a nonsymmetric channel. However
according to Li~\cite{li97:_molecular_evolution} mutation models with
many parameters are not necessarily accurate, due to the estimation
issues involved.

Under $m$ cascaded mutation stages we have a Markov chain $Y\to Z_{(1)}\to
Z_{(2)}\to\cdots\to Z_{(m)}$, and model~\eqref{eq:pi_symm} leads to the
overall transition probability matrix $\Pi^m$ between $Y$ and $Z_{(m)}$.
As $\Pi=\Pi^T$ we can write $\Pi^m=\mathrm{V}\, \mathrm{D}^m\,
\mathrm{V}^{T}$, with the eigenvalues of $\Pi$ arranged in a diagonal
matrix $\mathrm{D}\triangleq \diag(1,\,\lambda,\,\mu,\,\mu)$, where
\begin{eqnarray}
  \lambda&\triangleq& 1-\frac{4\gamma}{3} q\\
  \mu&\triangleq& 1-2\left(1-\frac{\gamma}{3}\right)q, 
\end{eqnarray}
and $\mathrm{V}$ a matrix whose columns are the normalised
eigenvectors of $\Pi$ associated to the corresponding eigenvalues in
$\mathrm{D}$, that is
 \begin{equation}
 \mathrm{V}=\frac{1}{2}\left[
   \begin{array}{r r r r}
     {\color{white}+}1&1&-\sqrt{2}&0\\
     1&-1&0&-\sqrt{2}\\
     1&-1&0&\sqrt{2}\\ 
     1&1&\sqrt{2}&0\\ 
  \end{array}\right].\label{eq:eigenv_matrix}
 \end{equation}
 From the diagonalisation of $\Pi^m$ it is straightforward to see that
 the elements of its diagonal all take the value $\frac{1}{4}\left(1+
   2 \mu^m+\lambda^m\right)$, the elements of its skew diagonal take
 the value $\frac{1}{4}\left(1- 2 \mu^m+\lambda^m\right)$, and the
 rest of its entries are
 $\frac{1}{4}\left(1-\lambda^m\right)$. Therefore any row (column) of
 this matrix contains the same probabilities, as $\Pi^m$ is also the
 transition matrix of a symmetric channel. From the diagonal elements
 one can see that the accumulated base substitution mutation rate
 after $m$ cascaded stages is given by
 \begin{equation}
   q^{(m)}=p(Z_{(m)}\neq y|Y=y)=1-\frac{1}{4}\left(1+ 2 \mu^m+\lambda^m\right).\label{eq:qm}
\end{equation}

When $q>0$, $\lim_{m\to\infty}
q^{(m)}|_{\gamma>0}=3/4$ but $\lim_{m\to\infty}
q^{(m)}|_{\gamma=0}=1/2$, because $|\mu|< 1$ for any $\gamma$ and
$|\lambda|<1$ when~$\gamma>0$, but $\lambda=1$ when $\gamma=0$. The
behaviour of this particular case is connected to the fact that we must have both
$q\in(0,1]$ and $\gamma\in(0,3/2]$ for the Markov chain to be
aperiodic and irreducible, and thus possess a limiting stationary
distribution. From the previous considerations, the limiting
distribution ---that is, the distribution of $Z_{(\infty)}$--- is uniform,
because $\lim_{m\to\infty} \Pi^m =\frac{1}{4}\mathbf{1}^T\mathbf{1}$.
When $\gamma=1$ and $q=3/4$ then $\Pi
=\frac{1}{4}\mathbf{1}^T\mathbf{1}$, and hence every $Z_{(m)}$ is
uniformly distributed as in the limiting case.

Lastly, under the base substitution mutation model that we are
considering, codons undergo a mutation channel modelled by the
$64\times 64$ transition probability matrix
\begin{equation}
  {\boldsymbol \Pi}=
  [p(\mathbf{Z}=\mathbf{z}|\mathbf{Y}=\mathbf{y})]=\Pi\otimes \Pi
  \otimes \Pi,\label{eq:kron}
\end{equation}
where $\otimes$ is the Kronecker product. This is because
$p(\mathbf{Z}=\mathbf{z}|\mathbf{Y}=\mathbf{y})=\prod_{i=1}^3
p(Z=z_i|Y=y_i)$ according to our memoryless channel assumption.
Trivially this channel is also symmetric. When $m$ mutation stages are
considered, $\Pi^m$ replaces $\Pi$ in~\eqref{eq:kron}, since
${\boldsymbol \Pi}^m=(\Pi\otimes \Pi \otimes \Pi)^m=\Pi^m\otimes \Pi^m
\otimes \Pi^m$~\cite{magnus99:matrix}.

\section{Capacity Analysis} 
\label{sec:capac-subs}

\subsection{Noncoding DNA} 
We will firstly consider this simple case, which will also establish a
basic upper bound to cDNA capacity. As per our discussion in
Section~\ref{sec:preliminary-concepts}, we are assuming that a
embedder can overwrite or append a host ncDNA strand~$\mathbf{x}^b$,
which amounts to freely choosing the input $\mathbf{y}^b$ to the
mutation channel. Therefore in this case the channel capacity is given
by $C_\mathrm{nc}\triangleq \max I(Z_{(m)};Y)$ bits/base, where the
maximisation is over all distributions of~$Y$. For the mutation model
considered, this capacity is that of the symmetric channel, in which
$H(Z_{(m)}|Y)$ is independent of the input and uniformly
distributed~$Y$ leads to uniformly distributed~$Z_{(m)}$. Hence
\begin{equation}
  \label{eq:cnc}
  C_\mathrm{nc} = \log{|\mathcal{X}|} - H(Z_{(m)}|Y) \textrm{ bits/base},
\end{equation}
where $H(Z_{(m)}|Y)=\sum_{z\in\mathcal{X}} p(Z_{(m)}=z|Y=y)\log p(Z_{(m)}=z|Y=y)$ for any
 $y\in\mathcal{X}$, that is, the entropy of any row of
$\Pi^m$. Therefore 
\begin{eqnarray}
  H(Z_{(m)}|Y)&=&-\frac{1}{4}\left(1+
  2\mu^m+\lambda^m\right)\log\left(\frac{1}{4}\left(1+
  2\mu^m+\lambda^m\right)\right)\nonumber\\ 
&&-\frac{1}{4}\left(1-
  2\mu^m+\lambda^m\right)\log\left(\frac{1}{4}\left(1-
  2\mu^m+\lambda^m\right)\right)\nonumber\\
&&-\frac{1}{2}\left(1-\lambda^m\right)\log \left(\frac{1}{4}\left(1-\lambda^m\right)\right).
 \label{eq:hzum}
\end{eqnarray}
As long as the Markov chain is aperiodic and irreducible then
$\lim_{m\to\infty} C_\mathrm{nc}=0$. The reason is that since the
limiting distribution is independent of $Y$, then
$\lim_{m\to\infty}H(Z_{(m)}|Y)=H(Z_{(\infty)})=\log |\mathcal{X}|$. It is
interesting to note that, under aperiodicity and irreducibility of the
Markov chain, this zero limiting capacity will also apply to models
more involved than~\eqref{eq:pi_symm}, such as those in which the
channel matrix is parametrised by up to~9 independent values. Lastly,
we also have that $C_\mathrm{nc}|_{\gamma=1,q=3/4}=0$, since in this
case $Z_{(m)}$ is always uniformly distributed.

As a function of $\gamma$ the ncDNA capacity is bounded as follows
\begin{equation}
  C_\mathrm{nc}|_{\gamma=1}\le C_\mathrm{nc}\le C_\mathrm{nc}|_{\gamma=0}.\label{eq:ineq_cnc}
\end{equation}
Although it can be shown  with some effort that these
inequalities always hold true, it is much simpler to prove them for
the range of interest $\gamma \le 1$ and $q\le 1/2$. The latter
condition implies that both $0\le\lambda\le 1$ and $0\le\mu \le
1$. For fixed $m$ and $q$, the maximum (respectively, minimum) of
$C_\mathrm{nc}$ over $\gamma$ corresponds to the minimum
(respectively, maximum) of the accumulated base mutation
rate~$q^{(m)}$. Differentiating~\eqref{eq:qm} we obtain $\partial
q^{(m)}/\partial \gamma=(m q/3) \left(\lambda^{m-1}-\mu^{m-1}\right)$.
Therefore $q^{(m)}$ is monotonically increasing when $\gamma\le 1$ (as
this corresponds to $\lambda\ge\mu$), and then its maximum in that
range occurs when $\gamma=1$ and its minimum when $\gamma=0$.

The upper bound can be written as
\begin{eqnarray}
  C_\mathrm{nc}|_{\gamma=0}=2-h\left(\frac{1}{2}+\frac{1}{2}(1-2q)^m\right).\label{eq:uppercnc}
\end{eqnarray}
Notice that $\lim_{m\to \infty} C_\mathrm{nc}|_{\gamma=0}=1$, that is,
the capacity limit is not zero when $\gamma=0$ because then the Markov
chain is reducible. This case cannot happen in practice since it would
imply that transversion mutations are impossible, but it illustrates
that the higher the transition/transversion ratio $\varepsilon$, the
higher the capacity.

Figures~\ref{fig:Cnc_1} and~\ref{fig:Cnc_2} show $C_\mathrm{nc}$ for
two different values of $q$ representative of extreme values of the base
substitution mutation range $q$ per replication found in different
living beings and different sections of
genomes~\cite{kunkel04:dna_replication}. We observe the validity of
the bounds~\eqref{eq:ineq_cnc} and the limiting behaviours
discussed. From these figures we can also empirically see that a rule-of-thumb capacity cut-off point is given by $m\sim 6/(5\gamma q)$.

\begin{figure}[t]
  \centering
  \psfrag{G1}[l][l]{\tiny \begin{turn}{30}$\gamma=1$\end{turn}}
  \psfrag{G0.1}[l][b]{\tiny \begin{turn}{30}$\gamma=10^{-1}$\end{turn}}
  \psfrag{G0.01}[l][b]{\tiny \begin{turn}{30}$\gamma=10^{-2}$\end{turn}}
  \psfrag{G0.001}[l][b]{\tiny \begin{turn}{30}$\gamma=10^{-3}$\end{turn}}
  \psfrag{G0}[l][l]{\tiny \begin{turn}{30}$\gamma=0$\end{turn}}
  \psfrag{m}[t][]{\small Cascaded mutation stages ($m$)}
  \psfrag{Cnc}[][l]{\small $C_\mathrm{nc}$, bits/base}
  \begin{minipage}{0.5\linewidth}
    \centering
    \includegraphics[width=7.5cm]{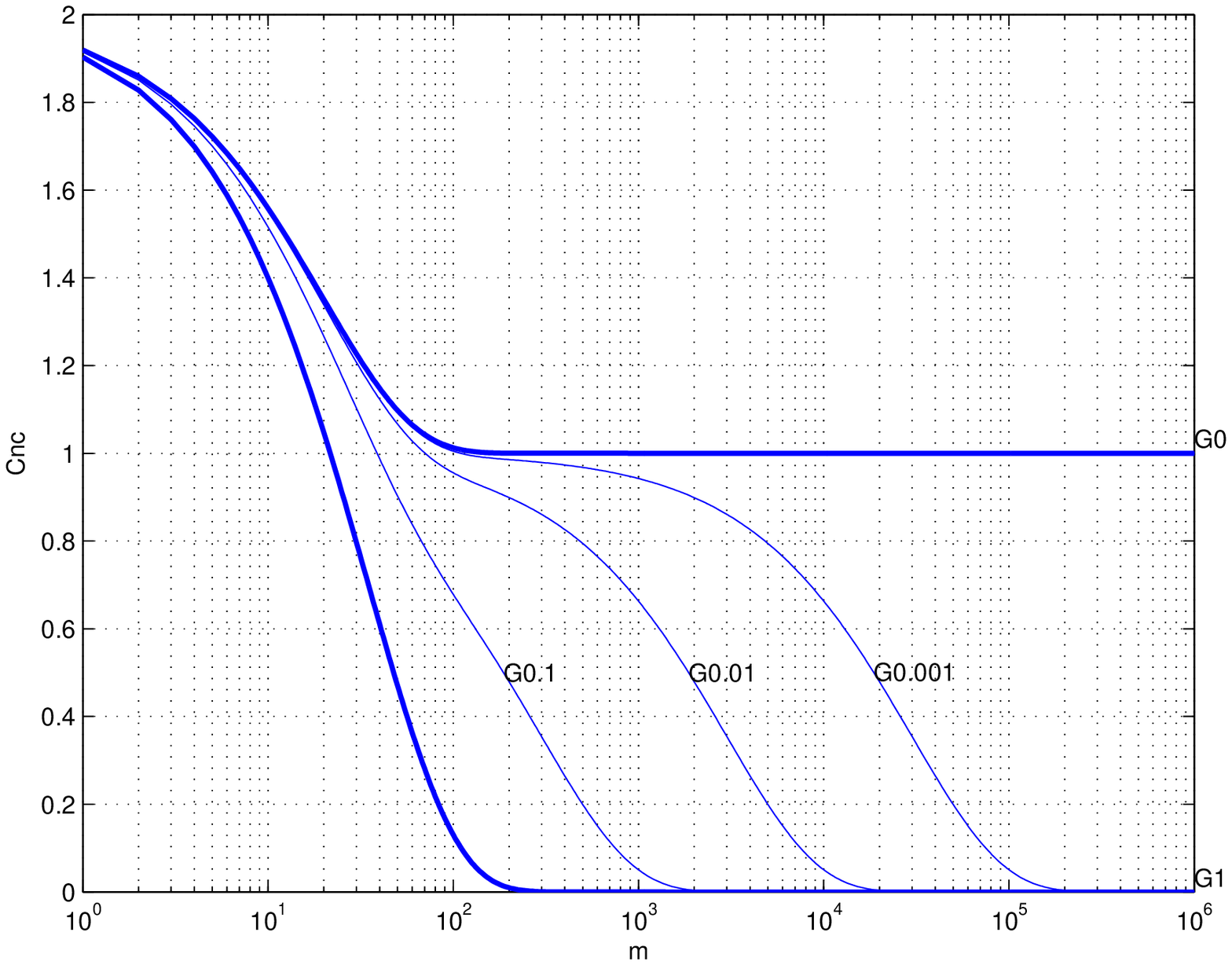}
    \caption{Embedding capacity in ncDNA ($q=10^{-2}$)}
    \label{fig:Cnc_1}
  \end{minipage}~\begin{minipage}{0.5\linewidth}
    \includegraphics[width=7.5cm]{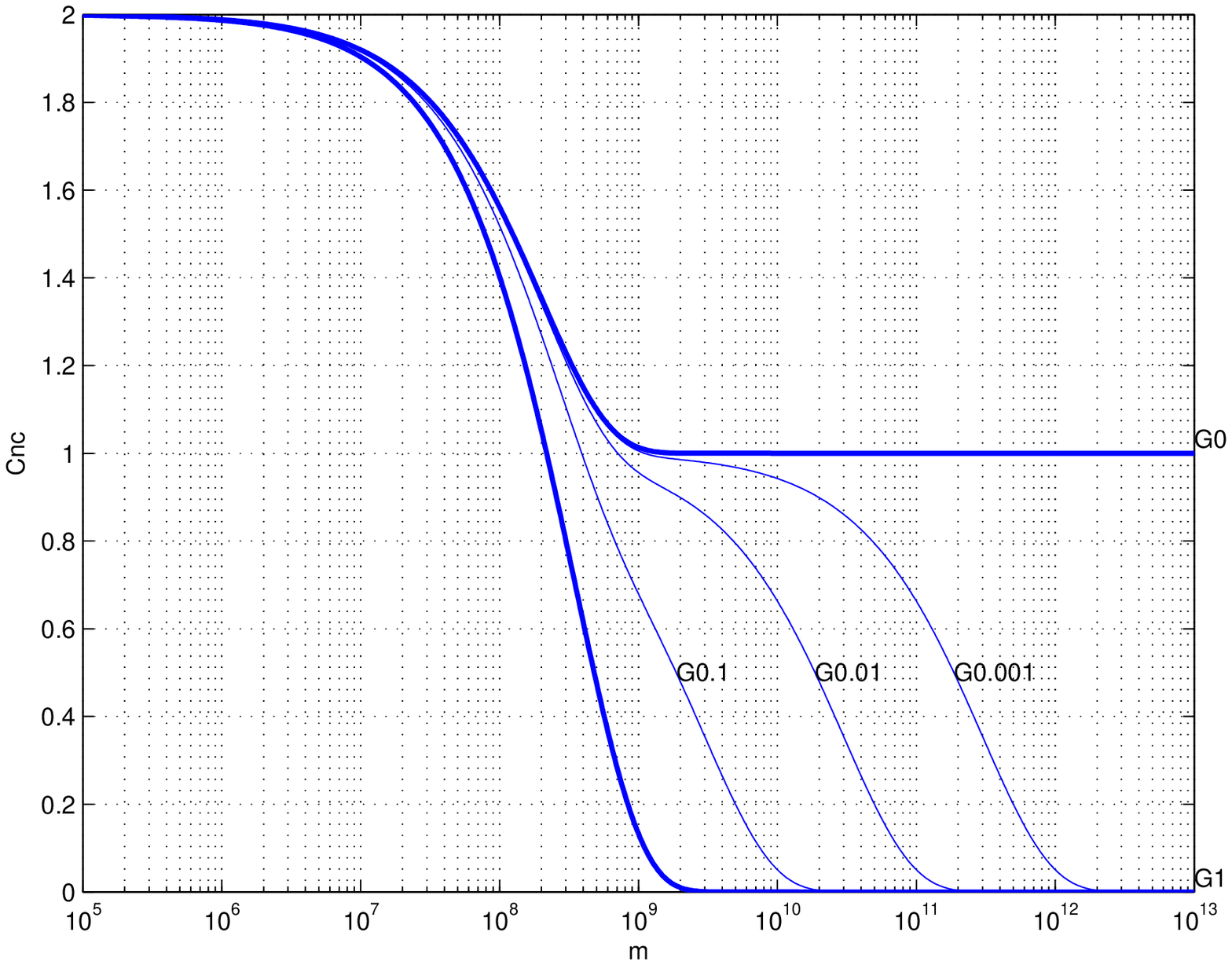}
    \caption{Embedding capacity in ncDNA ($q=10^{-9}$)}
    \label{fig:Cnc_2}
  \end{minipage} 
\end{figure}

\paragraph{Biological interpretations} 
We would like to point out that expression~\eqref{eq:cnc} also gives
the maximum mutual information between a DNA strand and its mutated
version in natural scenarios, independent of DNA data embedding
procedures. This has sometimes been termed the capacity of the genetic
channel in studies applying information theory to molecular
biology. Several authors have used the Jukes-Cantor model and
particular cases of the Kimura model ---apparently unaware of the
prior use of these models in molecular evolution studies--- in order
to estimate this capacity. The case $m=1$ was numerically evaluated by
May~\textit{et al.}~\cite{may03:_detection}, using values of $q$
estimated from different organisms and the Jukes-Cantor and Kimura
($\gamma=1/2$) models.  Some authors have also considered the
behaviour of capacity under cascaded mutation stages, that is,
$m>1$. Gutfraind~\cite{gutfraind06:error_tolerant} discussed the basic
effect of cascaded mutations on capacity (exponential decrease
with~$m$), although using a binary alphabet and the binary symmetric
channel. Both Battail~\cite{battail07:infotheory} and
May~\cite{may07:_bits_bases} computed capacity under cascaded mutation
stages using a quaternary alphabet and the Jukes-Cantor model. The
first author obtained his results analytically ---but using a
continuous-time approach rather than the discrete-time approach
followed here--- and the second one numerically. The results by
Battail are essentially consistent with the ones presented here
(similar capacity cut-off point), but the ones by May are not. The
capacity plots in~\cite{may07:_bits_bases} (taking the results for the
human genome) show a cut-off point of $m\approx 10^2$ for $q\approx
10^{-9}$, whereas $m\approx 10^{9}$ would have been expected according
to Figure~\ref{fig:Cnc_2}. Considering the extents of geological time,
where $m$ can easily reach $10^{9}$ and beyond, it seems clear that
the results in~\cite{may07:_bits_bases} underestimate capacity for
$m>1$.

In any case, none of the aforementioned approaches reflects the
capacity increase afforded by a mutation model allowing $\gamma<1$. It
is possible that the trend towards higher capacity observed as
$\gamma\to 0$ implies that evolution has favoured genetic building
blocks which feature an asymmetric behaviour under mutations (in our
case, pyrimidines versus purines instead of a hypothetically perfectly
symmetric set of four bases for which $\gamma=1$). If this assumption
is correct, this symmetry breaking must have occurred early in
evolutionary terms, since it is widely believed that the current
genetic machinery evolved from a former ``RNA
world''~\cite{Gilbert1986} in which life would only have been based on
the self-replicating and catalysing properties of RNA. In the RNA
world there would not have been translation to proteins, and therefore
no genetic code, and hence information was freely encoded using a
$4$-ary alphabet almost exactly like the one used in DNA. Note that
uracil, which replaces thymine in RNA, is also a pyrimidine, that is,
in the RNA world $\mathcal{Y}=\{\mathrm{U},\mathrm{C}\}$. With these
facts in mind, we may model the maximum transmissible information
under mutations in the RNA world by relying on~\eqref{eq:cnc}, and
thus see that the symmetry breaking conjecture above applies to the
evolution of RNA from predecessor genetic building blocks. We must
bear in mind that single-stranded molecules, such as RNA, are much
more mutation-prone than double-stranded ones such as DNA\footnote{For
  instance, RNA viruses such as HIV are known to exhibit base mutation
  rates of up to $10^{-2}$ per
  year~\cite{fu01:estimating}.}. Therefore smaller values of $m$ would
have sufficed for some type of symmetry breaking to be relevant in
terms of information transmission at early stages of life.

\subsection{Coding DNA}  
Unlike in the ncDNA case, embedding information in cDNA is a problem
of coding with side information at the encoder. Given a host sequence
$\mathbf{x}^c$, the encoder has to modify this host to produce an
information-carrying sequence $\mathbf{y}^c$ which must also encode
the same primary structure as $\mathbf{x}^c$ according to the genetic
code. This is equivalent to hiding data in a discrete host under an
embedding constraint. Nevertheless, apart from the trivial difference
of using a $4$-ary instead of typically a $2$-ary alphabet, several
issues set apart cDNA data embedding as a special problem. In order to
illustrate these issues consider momentarily a typical data hiding
scenario in which a discrete binary host, that is
$\mathbf{x}=[x_1,\cdots,x_n]$ with $x_i\in \mathcal{X}=\{0,1\}$, is
modified to embed a message~$m$ from a certain alphabet. The
watermarked signal $\mathbf{y}=e(\mathbf{x},m)$ must be close to
$\mathbf{x}$, where closeness is usually measured by means of the
Hamming distance $d_H(\mathbf{y},\mathbf{x})$. Pradhan \textit{et
  al.}~\cite{pradhan03:duality} and Barron \textit{et
  al.}~\cite{Barron03} have determined the achievable rate in this
scenario, assuming that the elements of $\mathbf{X}$ are uniformly
distributed, using the average distortion constraint
$\frac{1}{n}E[d_H(\mathbf{Y},\mathbf{X})]\le d$, and supposing that
$\mathbf{y}$ undergoes a memoryless binary symmetric channel with
crossover probability $q$. Their result is
 \begin{equation*}
   R^\mathrm{unif}=\mathrm{u.c.e.}\{h(d)-h(q)\} \mathrm{\; bits/host\; symbol},
 \end{equation*}
 where $\mathrm{u.c.e}\{\cdot\}$ is the upper concave envelope.
 Similarly, our initial goal for cDNA data embedding is obtaining the
 achievable rate for a fixed distribution of $X'=\alpha(\mathbf{X})$
 under the symmetric channel discussed in
 Section~\ref{sec:robustn-under-mutat}, in particular when
 $\mathbf{X}$ is uniformly distributed as in the analyses of Pradhan
 \textit{et al.}~\cite{pradhan03:duality} and Barron \textit{et
   al.}~\cite{Barron03}. Furthermore we will also obtain capacity,
 that is, the maximum achievable rate over all distributions of the
 host $X'$.

The first important difference in the cDNA data embedding scenario is
that average inequality constraints on the Hamming distance
---such as the ones used in~\cite{pradhan03:duality,Barron03}--- are
meaningless if one wants to carry through to $\mathbf{y}^c$ the full
biological functionality of~$\mathbf{x}^c$.  Instead, since it must
always hold that $\alpha(\mathbf{y}^c)=\alpha(\mathbf{x}^c)$, one must
establish the deterministic constraint
\begin{equation}\label{eq:constraint}
  d_H(\mathbf{y}',\mathbf{x}')=\sum_{i=1}^nd_H(y_i',x_i')=0.
\end{equation}
This requires that $d_H(y_i',x_i')=0$ for all $i=1,\cdots,n$. 

The second distinguishing feature of cDNA data embedding is due to the
variable support of the channel input variable. Whereas in discrete
data hiding with binary host one always has that $y_i\in\{0,1\}$
independently of $x_i$, in cDNA data embedding we have that
$\mathbf{y}_i\in\mathcal{S}_{\alpha(\mathbf{x}_i)}$ so that the
constraint~\eqref{eq:constraint} can always be satisfied. Therefore
the support of $\mathbf{y}_i$ is dependent on~$\mathbf{x}_i$, as codon
equivalence is not evenly spread over the ensemble of amino acids (see
Table~\ref{tab:amino}).

\subsubsection{Achievable Rate}\label{sec:achi-rate-analys}
Since side information at the encoder must be taken into account in
the cDNA case, then the achievable rate is given by Gel'fand and
Pinsker's formula~\cite{Gelfand} $R_\mathrm{c}^{X'}=\max
I(\mathbf{Z}_{(m)};\mathbf{U})-I(X';\mathbf{U})$ bits/codon, where the
maximisation is for nonnegative values of the functional on all
distributions $p(\mathbf{y},\mathbf{u}|x')$ under the constraint
$d_H(\alpha(\mathbf{y}),x')=0$, with~$\mathbf{U}$~an auxiliary random
variable that we will discuss next. Note that $R_\mathrm{c}^{X'}$
represents the maximum achievable rate when the host cDNA amino acid
sequence is distributed as $X'$.

Gel'fand and Pinsker showed in~\cite{Gelfand} that in the maximisation
problem above one may assume that the channel input is a deterministic
function of the side information $X'$ and the auxiliary variable
$\mathbf{U}$, that is, $\mathbf{Y}=e(X',\mathbf{U})$. Since the
support of $\mathbf{Y}|x'$ must be the set of codons
$\mathcal{S}_{x'}$ corresponding to amino acid $x'$ ---so that the
biological constraint can always be satisfied--- then the cardinality
of the support of $\mathbf{U}|x'$ has to coincide with the
multiplicity of $x'$, that is,~$|\mathcal{S}_{x'}|$. The support of
$\mathbf{U}|x'$ must actually be $\mathcal{S}_{x'}$, because
$\mathbf{U}$ must also act as a good source code for $X'$ in order to
minimise $I(X';\mathbf{U})$ under the genetic constraint, and if the
support of $\mathbf{U}|x'$ is otherwise then the constraint cannot always be
met. One can now establish $\mathbf{Y}|x'=\mathbf{U}|x'$ without loss
of generality, although any permutation of the elements of
$\mathcal{S}_{x'}$ is actually valid to define
$\mathbf{Y}|x'=e(x',\mathbf{U})$. Therefore in the following one may
consider that $\mathbf{Y}=\mathbf{U}$, that is, that $\mathbf{U}$ is
the mutation channel input. Noticing that $\mathcal{S}_{x'}\cap
\mathcal{S}_{y'}=\emptyset$ for $x'\neq y'\in\mathcal{X'}$, the
distribution of $\mathbf{U}$ can be put as
$p(\mathbf{u})=p(\mathbf{u}|x')p(x')$ when
$\mathbf{u}\in\mathcal{S}_{x'}$. This discussion on $\mathbf{U}$ also
implies that $H(X'|\mathbf{U})=0$, since given a codon $\mathbf{u}$
there is no uncertainty on the amino acid represented, and therefore
$I(X';\mathbf{U})= H(X')$

Since $\mathbf{Y}|(x',\mathbf{u})$ is deterministic, from the
considerations above we have that the achievable rate for a fixed
distribution of $X'$ is given by
\begin{equation}
  \label{eq:achievable_rate}
  R_\mathrm{c}^{X'}=\max_{p(\mathbf{u}|x')}
  I(\mathbf{Z}_{(m)};\mathbf{U})-H(X')\text{ bits/codon}.
\end{equation}
As $H(\mathbf{Z}_{(m)}|\mathbf{U})$ only depends on the transition
probabilities of the symmetric channel, and as trivially $H(X')$ only depends
on~$X'$,~\eqref{eq:achievable_rate} amounts to the constrained
maximisation of $H(\mathbf{Z}_{(m)})$. 

There are several cases in which~\eqref{eq:achievable_rate} can be
analytically determined, which are discussed next. First of all, since
$C_\mathrm{nc}|_{\gamma=1,q=3/4}=0$ then
$R_\mathrm{c}^{X'}|_{\gamma=1,q=3/4}=0$ for any~$X'$, because
$R_\mathrm{c}^{X'}\le 3C_\mathrm{nc}$. Therefore in this catastrophic
case the choice of $p(\mathbf{u}|x')$ is irrelevant. Furthermore it
can be shown that $p(\mathbf{u}|x')=1/|\mathcal{S}_{x'}|$, that is,
$\textbf{U}|x'$ uniformly distributed, is the maximising strategy in
two situations, which are discussed in the following lemmas.
\begin{lem}\label{lem:rq0}
  If $q=0$ then the achievable rate is
  \begin{eqnarray}
    \label{eq:rate_q0}
    R_\mathrm{c}^{X'}|_{q=0} &=&E\left[\log |\mathcal{S}_{X'}|\right]\;\mathrm{ bits/codon}.
  \end{eqnarray}
\end{lem}
\begin{proof}
  Using the chain rule of the entropy we can write
  $H(\mathbf{U},X')=H(\mathbf{U})+H(\mathbf{U}|X')=H(X')+H(X'|\mathbf{U})$.  
  As $H(X'|\mathbf{U})=0$, and as $\mathbf{Z}_{(m)}=\mathbf{U}$ when
  $q=0$, then the achievable rate is given by
  $R_\mathrm{c}^{X'}|_{q=0}=\max_{p(\mathbf{u}|x')}H(\mathbf{U})-H(X')=\max_{p(\mathbf{u}|x')}H(\mathbf{U}|X')$. We
  just need to see now that
  $H(\mathbf{U}|X')=\sum_{x'\in\mathcal{X}'}p(x') H(\mathbf{U}|x')$
  is maximised when $H(\mathbf{U}|x')$ is maximum for all $x'$, which
  implies that $\mathbf{U}|x'$ be uniformly distributed in all
  cases. Then $H(\mathbf{U}|x')=\log |\mathcal{S}_{x'}|$
  and~\eqref{eq:rate_q0} follows.
\end{proof}

  \noindent\textbf{Remark.} Note that~\eqref{eq:rate_q0} is the embedding rate intuitively
  expected in the mutation-free case. For example, if~$\mathbf{X}$
  were uniformly distributed, which would yield
  $X'=\alpha(\mathbf{X})$ nonuniform with pmf
  $p(x')=|\mathcal{S}_{x'}|/|\mathcal{X}|^3$, then we would obviously
  compute the rate as
  $R_\mathrm{c}^{\alpha(\textrm{unif})}|_{q=0}=\sum_{x'}\frac{|\mathcal{S}_{x'}|}{|\mathcal{X}|^3}\log|\mathcal{S}_{x'}|=
  1.7819$ bits/codon, since $|\mathcal{S}_{x'}|$ choices are available
  to the embedder when the host amino acid is $x'$. The rate in the
  uniform case can actually be obtained in closed form for every $q$
  using the following result.
  
\begin{lem}\label{lem:unif}
  If $\mathbf{X}$ is uniformly distributed then the achievable rate is
  \begin{equation}
    \label{eq:rate_unif}
    R_\mathrm{c}^{\alpha(\mathrm{unif})}=\widetilde{C}_\mathrm{nc}-H(X')\;\mathrm{bits/codon},
  \end{equation}
  where $\widetilde{C}_\mathrm{nc}\triangleq \max I(\mathbf{Z}_{(m)};\mathbf{U})$
  and this maximisation is unconstrained on $p(\mathbf{u})$, that is,
  $\widetilde{C}_\mathrm{nc}$ is the capacity of the symmetric codon mutation
  channel. 

\end{lem}
\begin{proof}
  Since $p(\mathbf{u})=p(\mathbf{u}|x')p(x')$ when
  $\mathbf{u}\in\mathcal{S}_{x'}$, with uniformly distributed
  $\mathbf{X}$ we have that
  $p(\mathbf{u})=p(\mathbf{u}|x')|\mathcal{S}_{x'}|/|\mathcal{X}|^3$
  when $\mathbf{u}\in\mathcal{S}_{x'}$. Therefore choosing
  $\mathbf{U}|x'$ to be uniformly distributed implies that
  $p(\mathbf{u})=1/|\mathcal{X}|^3$ for all $\mathbf{u}$.  Since
  ${\boldsymbol \Pi}^m$ is symmetric and a uniform input maximises
  mutual information over a symmetric channel, then
  $\widetilde{C}_\mathrm{nc}=\max I(\mathbf{Z}_{(m)};\mathbf{U})$ is
  achieved in~\eqref{eq:achievable_rate}.
\end{proof}

  \noindent\textbf{Remarks.} Since  $\widetilde{C}_\mathrm{nc}|_{q=0}=\log
  |\mathcal{X}|^3$, observe that the particular case in the previous
  remark can be written as well as
  $R_\mathrm{c}^{\alpha(\textrm{unif})}|_{q=0}=\log|\mathcal{X}|^3
  -H(X')$.  An interesting insight is also afforded by seeing that the
  three parallel symmetric channels undergone by the bases in a codon
  are mutually independent, and hence one can use the equality
  $\widetilde{C}_\mathrm{nc}=3\, C_\mathrm{nc}$
  in~\eqref{eq:rate_unif}. As $H(X')$ is the lower bound to the
  lossless source coding rate of $X'$, expression~\eqref{eq:rate_unif}
  tells a fact that is intuitively appealing but which is only exact
  when~$\mathbf{X}$ is uniform: the cDNA embedding rate is the same as
  three times the ncDNA embedding rate minus the rate needed to
  losslessly convey the primary structure of the host to the decoder.

  Unlike in the case considered above, the distribution of
  $\mathbf{X}$ in real cDNA sequences (that is, genes) is not uniform.
  To start with, there can only be a single \textit{Stp} codon in a
  sequence that encodes a protein (gene). As in many other channel
  capacity problems, it does not seem possible in general to
  analytically derive the optimum set of pmf's $p(\mathbf{u}|x')$ in
  order to compute the achievable rate $R_\mathrm{c}^{X'}$
  corresponding to a host distributed as~$X'$. To see why one can pose
  the analytical optimisation problem and see that it involves solving
  a nontrivial system of $|\mathcal{X}|^3+|\mathcal{X}'|$ nonlinear
  equations and unknowns.  However the numerical solution is
  straightforward by means of the Blahut-Arimoto
  algorithm~\cite{blahut72:_comput} adapted to the side-informed
  scenario. Such an algorithm has been described by Dupuis~\textit{et
    al.}~\cite{dupuis04:_blahut}. An example of the optimal
  distributions numerically obtained for a particular case is shown in
  Figure~\ref{fig:puxp}.

\begin{figure}[t]
  \centering
  \psfrag{Ala}[][]{\color{black}\begin{turn}{90}\tiny Ala\end{turn}}
  \psfrag{Arg}[][]{\color{black}\begin{turn}{90}\tiny Arg\end{turn}}
  \psfrag{Asn}[][]{\color{black}\begin{turn}{90}\tiny Asn\end{turn}}
  \psfrag{Asp}[][]{\color{black}\begin{turn}{90}\tiny Asp\end{turn}}
  \psfrag{Cys}[][]{\color{black}\begin{turn}{90}\tiny Cys\end{turn}}
  \psfrag{Gln}[][]{\color{black}\begin{turn}{90}\tiny Gln\end{turn}}
  \psfrag{Glu}[][]{\color{black}\begin{turn}{90}\tiny Glu\end{turn}}
  \psfrag{Gly}[][]{\color{black}\begin{turn}{90}\tiny Gly\end{turn}}
  \psfrag{His}[][]{\color{black}\begin{turn}{90}\tiny His\end{turn}}
  \psfrag{Ile}[][]{\color{black}\begin{turn}{90}\tiny Ile\end{turn}}
  \psfrag{Leu}[][]{\color{black}\begin{turn}{90}\tiny Leu\end{turn}}
  \psfrag{Lys}[][]{\color{black}\begin{turn}{90}\tiny Lys\end{turn}}
  \psfrag{Met}[][]{\color{black}\begin{turn}{90}\tiny Met\end{turn}}
  \psfrag{Phe}[][]{\color{black}\begin{turn}{90}\tiny Phe\end{turn}}
  \psfrag{Pro}[][]{\color{black}\begin{turn}{90}\tiny Pro\end{turn}}
  \psfrag{Ser}[][]{\color{black}\begin{turn}{90}\tiny Ser\end{turn}}
  \psfrag{Thr}[][]{\color{black}\begin{turn}{90}\tiny Thr\end{turn}}
  \psfrag{Trp}[][]{\color{black}\begin{turn}{90}\tiny Trp\end{turn}}
  \psfrag{Tyr}[][]{\color{black}\begin{turn}{90}\tiny Tyr\end{turn}}
  \psfrag{Val}[][]{\color{black}\begin{turn}{90}\tiny Val\end{turn}}
  \psfrag{Stp}[][]{\color{black}\begin{turn}{90}\tiny \textit{Stp}\end{turn}}
  \psfrag{puxp}[b][]{$p(\mathbf{u}|x')$}
  \psfrag{x}[][]{$x'$}
  \includegraphics[width=10cm, height=8cm]{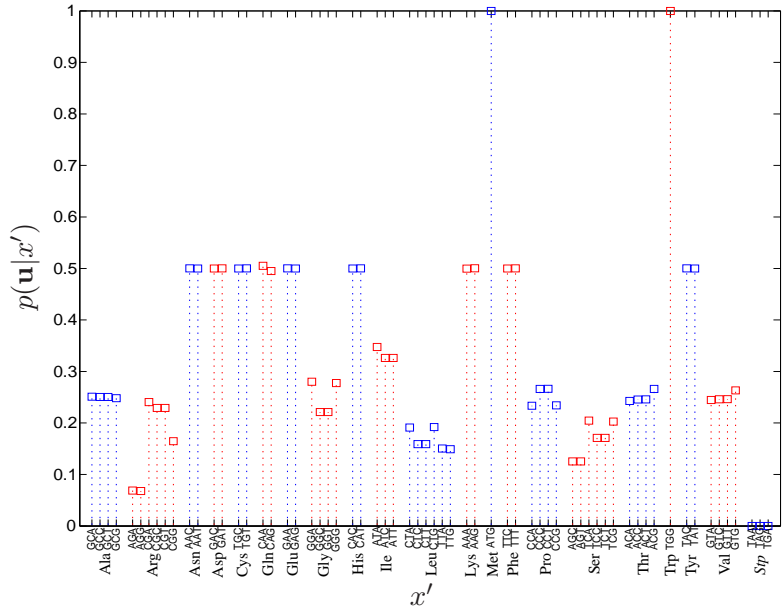}
  \caption{Example of maximising $p(\mathbf{u}|x')$ distributions
    numerically obtained using the Blahut-Arimoto algorithm and
    $p(x')$ corresponding to gene Ypt7 from yeast (GenBank accession
    number NC\_001145), employing $\gamma=0.1$, $q=10^{-2}$,
    $m=10$. Conditional pmf's are depicted in alternating red and blue
    colours to facilitate plot reading.}
  \label{fig:puxp}
\end{figure}

A last observation is that, in general,
$p(\mathbf{u}|x')=1/|\mathcal{S}_{x'}|$ turns out to yield a good
approximation to the exact numerical solution. Note that for
distributions of $X'$ different from the one in Lemma~\ref{lem:unif} one cannot
produce a uniform input $\mathbf{U}$ to the symmetric channel in order
to generate a uniform output $\mathbf{Z}_{(m)}$. This is the case
illustrated in Figure~\ref{fig:puxp}; nonetheless, note from this
figure that $p(\mathbf{u}|x')$ does not differ excessively from a
uniform distribution for several amino acids $x'$.  A justification of this
behaviour is as follows. Using the fact that conditioning cannot
increase entropy, a suboptimal maximisation approach is given by
maximising the lower bound
$H(\mathbf{Z}_{(m)}|X')=\sum_{x'\in\mathcal{X}'}p(x') H(\mathbf{Z}_{(m)}|x')
\le H(\mathbf{Z}_{(m)})$.  This requires maximising
$H(\mathbf{Z}_{(m)}|x')=-\sum_{\mathbf{z}\in\mathcal{X}^3}
p(\mathbf{z}|x')\log p(\mathbf{z}|x')$ for all $x'$. Observing from
Table~\ref{tab:amino} that  codons mapping to the same
amino acid share in many cases up to two bases, we can approximate
$p(\mathbf{z}|x')\approx 0$ when $\mathbf{z}\notin \mathcal{S}_{x'}$
and $p(\mathbf{z}|x')\approx
\left(\sum_{\mathbf{v}\in\mathcal{S}_{x'}}
  p(\mathbf{z}|\mathbf{v})\right)^{-1}
\sum_{\mathbf{u}\in\mathcal{S}_{x'}} p(\mathbf{z}|\mathbf{u})
p(\mathbf{u}|x')$ when $\mathbf{z}\in \mathcal{S}_{x'}$. With this
approximation, whenever $\sum_{\mathbf{u}\in\mathcal{S}_{x'}}
p(\mathbf{z}|\mathbf{u})$ is constant for all
$\mathbf{z}\in\mathcal{S}_{x'}$, choosing $p(\mathbf{u}|x')$ to be
uniform implies that $p(\mathbf{z}|x')$ is also uniform, which
maximises $H(\mathbf{Z}_{(m)}|x')$. It can be verified that this condition
holds for all $x'$ such that $|\mathcal{S}_{x'}|=1,2,4$, which
accounts for 16 out of the 21 elements in $\mathcal{X}'$.

Figures~\ref{fig:Rg1_1}-\ref{fig:Rg01_2} present the achievable rates
for several distributions of $X'$. Shown are the rates for the
distributions corresponding to two real genes: Ypt7
(\textit{S. Cerevisiae}) and FtsZ (\textit{B. Subtilis}), whose
GenBank accession numbers are NC\_001145 and NC\_000964, respectively;
also depicted are the rate~\eqref{eq:rate_unif} for $\mathbf{X}$
uniform and the rate for the deterministic distribution of $X'$ with
outcome $\mathrm{Ser}$, which, as we will discuss in
Section~\ref{sec:capacity-analysis}, yields capacity. We observe in
these plots that there is barely a difference in the results obtained
with the Blahut-Arimoto algorithm and the uniform approximation to
$p(\mathbf{u}|x')$ that we have discussed.
\begin{figure}[t]
  \centering
  \psfrag{a}[l][l]{\tiny Blahut-Arimoto}
  \psfrag{b}[l][l]{\tiny $p(\mathbf{u}|x')=1/|\mathcal{S}_{x'}|$}
  \psfrag{ftsz}[l][l]{\tiny \begin{turn}{30}ftsZ\end{turn}}
  \psfrag{ypt7}[l][l]{\tiny \begin{turn}{30}ypt7\end{turn}}
  \psfrag{unif}[l][l]{\tiny \begin{turn}{30}$\alpha(\mathrm{unif})$\end{turn}}
  \psfrag{ser}[l][l]{\tiny \begin{turn}{30}Ser\end{turn}}
  \psfrag{m}[t][]{\small Cascaded mutation stages ($m$)}
  \psfrag{RR}[][l]{\small $R_\mathrm{c}^{X'}$, bits/codon}
  \begin{minipage}{0.5\linewidth}
    \centering
    \includegraphics[width=7.5cm]{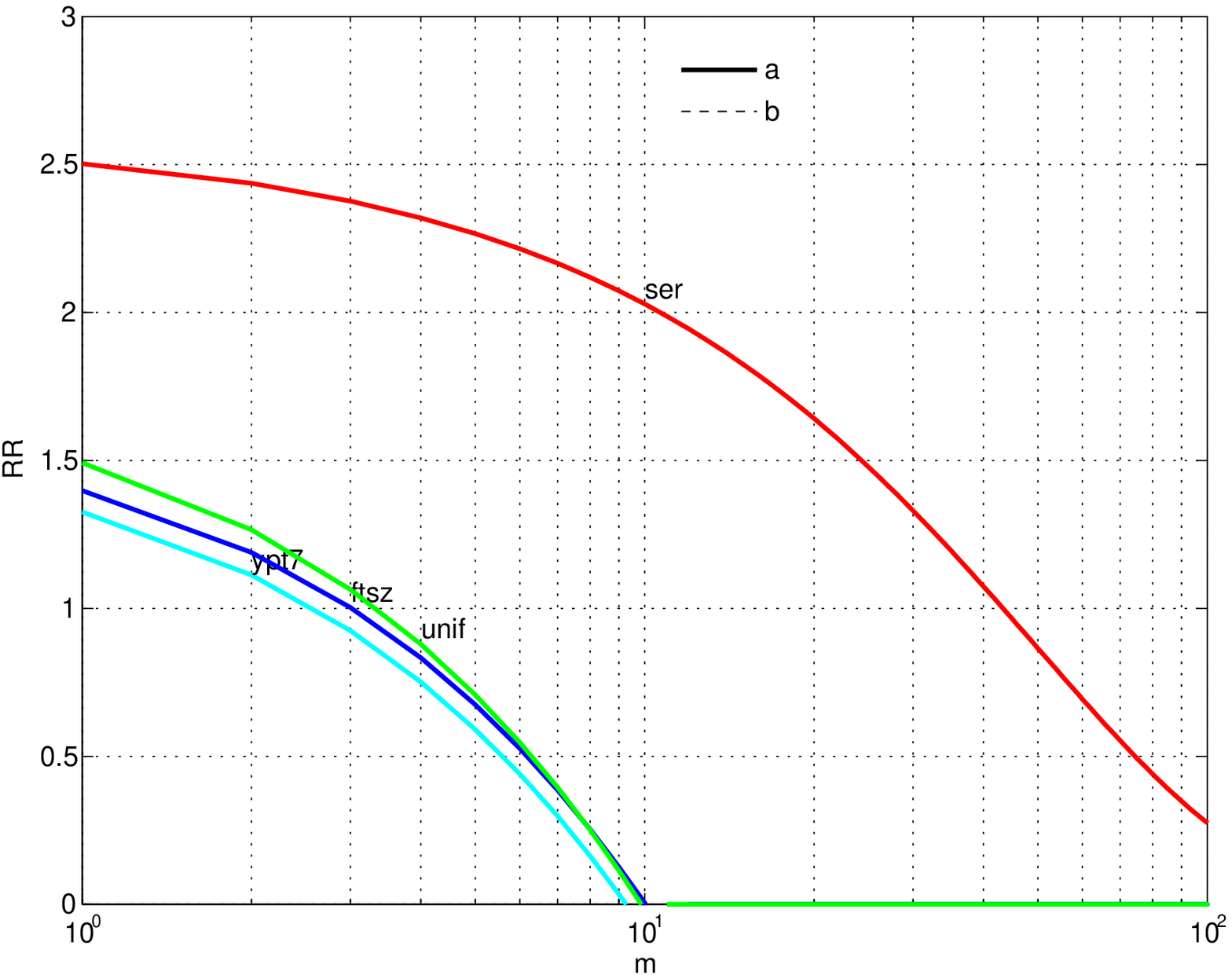}
    \caption{Embedding rate in cDNA for different distributions of
      $X'$ ($\gamma=1, q=10^{-2}$)}
    \label{fig:Rg1_1}
  \end{minipage}~\begin{minipage}{0.5\linewidth}
    \includegraphics[width=7.5cm]{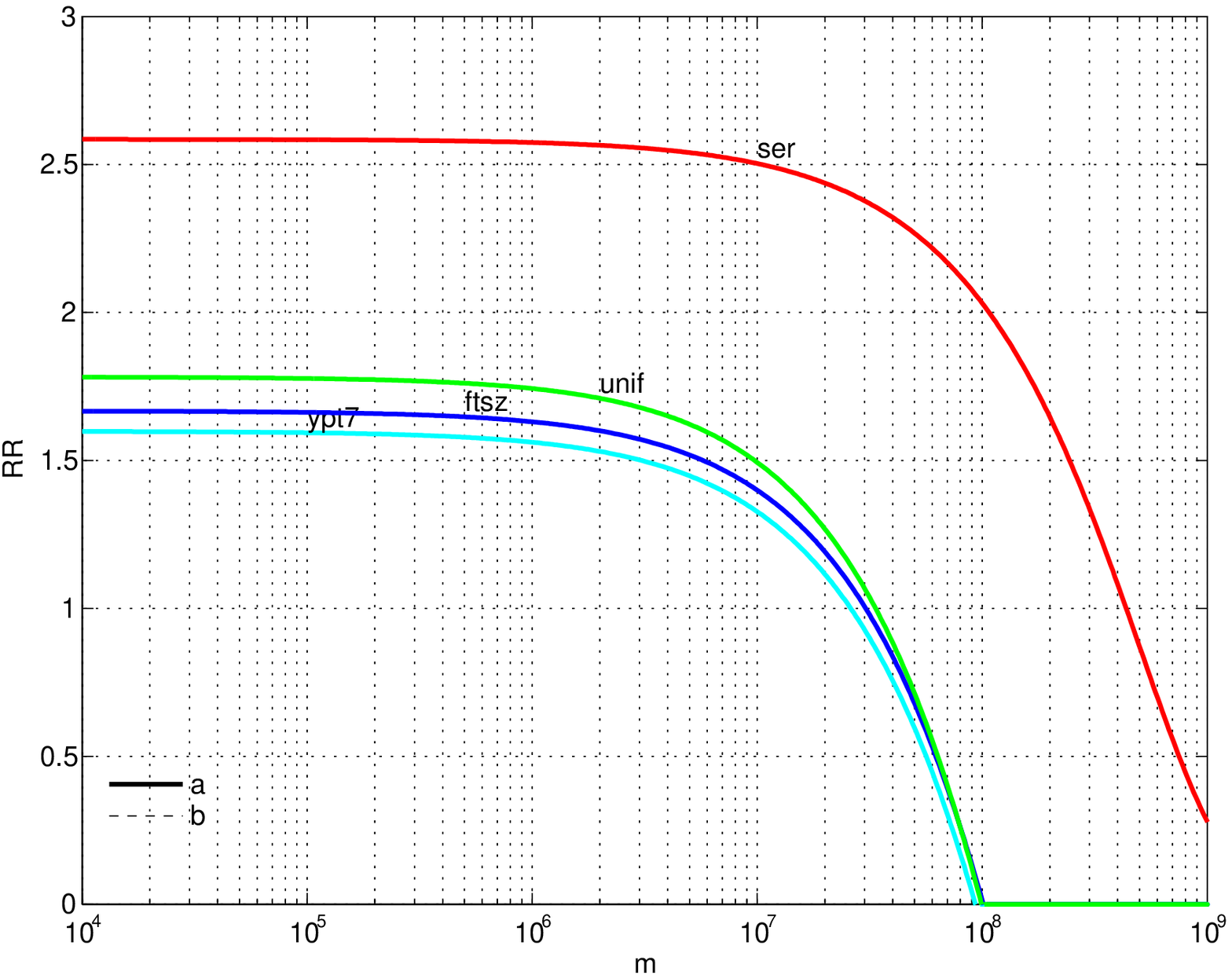}
    \caption{Embedding rate in cDNA for different  distributions of
      $X'$ ($\gamma=1, q=10^{-9}$)}
    \label{fig:Rg1_2}
  \end{minipage} 
\end{figure}

\begin{figure}[t]
  \centering
  \psfrag{a}[l][l]{\tiny Blahut-Arimoto}
  \psfrag{b}[l][l]{\tiny $p(\mathbf{u}|x')=1/|\mathcal{S}_{x'}|$}
  \psfrag{ftsz}[l][l]{\tiny \begin{turn}{30}ftsZ\end{turn}}
  \psfrag{ypt7}[l][l]{\tiny \begin{turn}{30}ypt7\end{turn}}
  \psfrag{unif}[l][l]{\tiny \begin{turn}{30}$\alpha(\mathrm{unif})$\end{turn}}
  \psfrag{ser}[l][l]{\tiny \begin{turn}{30}Ser\end{turn}}
  \psfrag{m}[t][]{\small Cascaded mutation stages ($m$)}
  \psfrag{RR}[][l]{\small $R_\mathrm{c}^{X'}$, bits/codon}
  \begin{minipage}{0.5\linewidth}
    \centering
    \includegraphics[width=7.5cm]{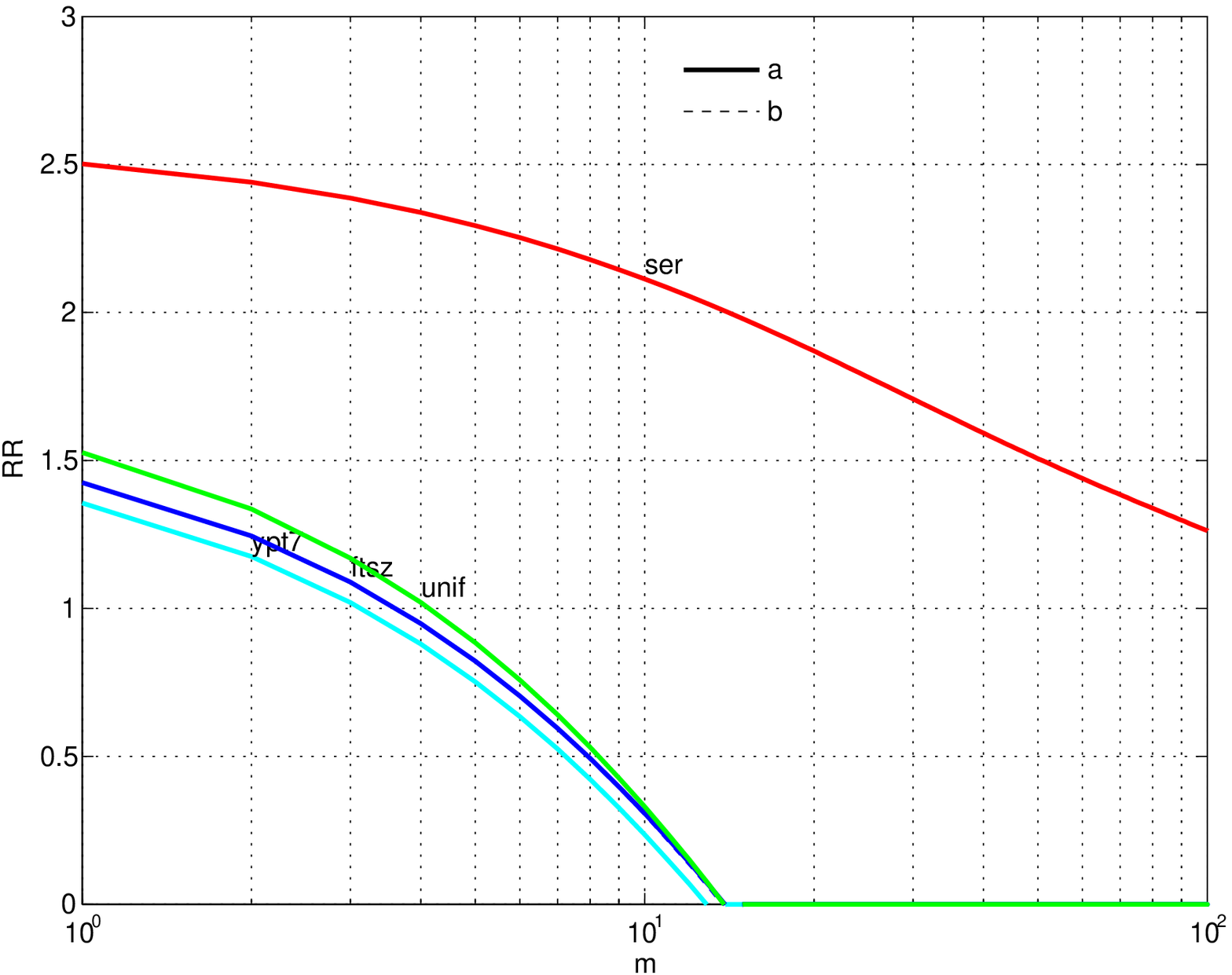}
    \caption{Embedding rate in cDNA for different distributions of
      $X'$ ($\gamma=0.1, q=10^{-2}$)}
    \label{fig:Rg01_1}
  \end{minipage}~\begin{minipage}{0.5\linewidth}
    \includegraphics[width=7.5cm]{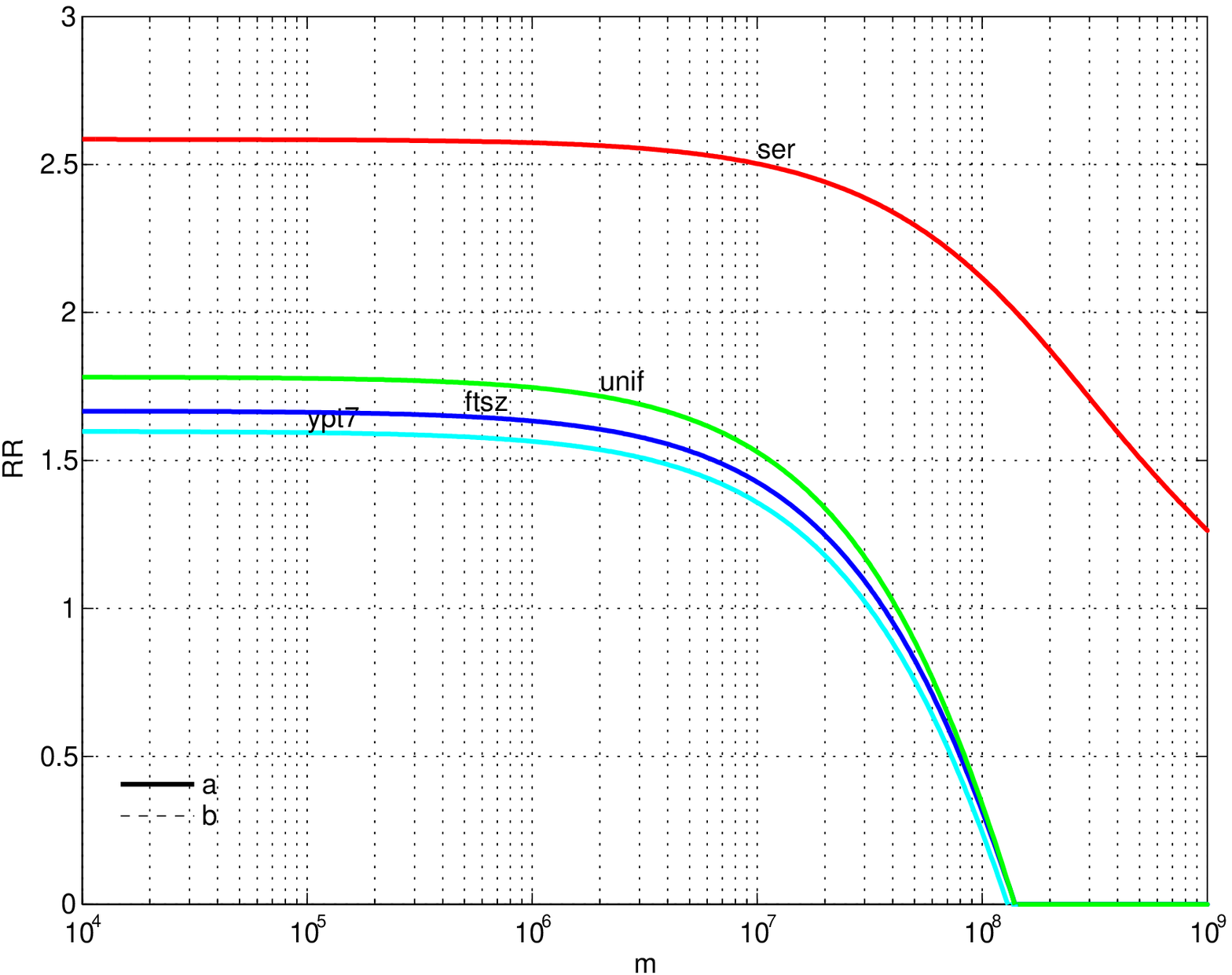}
    \caption{Embedding rate in cDNA for different distributions of
      $X'$ ($\gamma=0.1, q=10^{-9}$)}
    \label{fig:Rg01_2}
  \end{minipage} 
\end{figure}
   
\paragraph{Codon statistics preservation} If we require that the
original codon statistics of the host are preserved in the
information-carrying sequence, then we must peg $p(\mathbf{u}|x')$ to
the corresponding distribution of the host. Therefore in this case no
maximisation is required, and the corresponding rate will be lower or
equal than the one achieved without codon statistics
preservation. This type of constraint is equivalent to a
steganographic constraint in data hiding, since the pmf of the host is
preserved in the information-carrying sequence. A comparison of
maximum rates and codon statistics preservation rates for the same
genes as before is given in Figure~\ref{fig:R_steg}. Note that in the
uniform case both rates coincide because of Lemma~\ref{lem:unif}.

\begin{figure}[t]
  \centering
  \psfrag{a}[l][l]{\tiny Codon statistics preservation
    rate}
  \psfrag{b}[l][l]{\tiny Maximum rate (Blahut-Arimoto)}
  \psfrag{ftsz}[l][l]{\tiny \begin{turn}{30}ftsZ\end{turn}}
  \psfrag{ypt7}[l][l]{\tiny \begin{turn}{30}ypt7\end{turn}}
  \psfrag{unif}[l][l]{\tiny \begin{turn}{30}$\alpha(\mathrm{unif})$\end{turn}}
  \psfrag{m}[t][]{\small Cascaded mutation stages ($m$)}
  \psfrag{RR}[][l]{\small $R_\mathrm{c}^{X'}$, bits/codon}
  \centering
  \includegraphics[width=7.5cm]{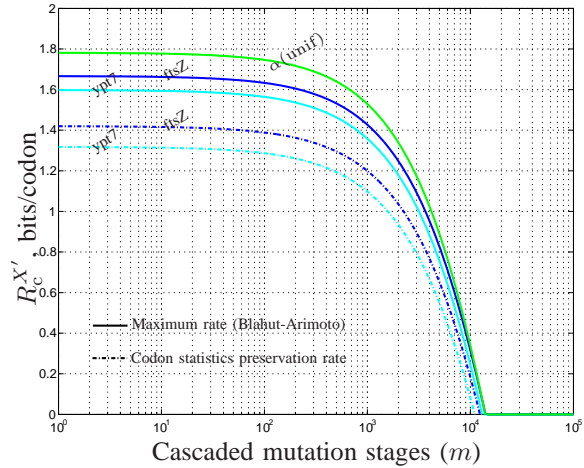}
  \caption{Comparison of cDNA embedding rate with and without
    codon statistics preservation constraints, for different distributions of $X'$
    ($\gamma=0.1, q=10^{-5}$)}
  \label{fig:R_steg}
\end{figure}

\subsubsection{Capacity}\label{sec:capacity-analysis}
  It remains the computation of capacity, that is,
  \begin{equation}
    \label{eq:capacity}
    C_\mathrm{c}=\max_{p(x')}\, R_\mathrm{c}^{X'}\text{ bits/codon}.
  \end{equation}
  It is simple to explicitly obtain $C_\mathrm{c}$ in two particular
  cases discussed in Section~\ref{sec:achi-rate-analys}.
  \begin{itemize}
  \item $\gamma=1$ and $q=3/4$: obviously,
    $C_\mathrm{c}|_{\gamma=1,q=3/4}=R_\mathrm{c}^{X'}|_{\gamma=1,q=3/4}=0$. However,
    the point that has to be made here is that, although this is true
    for any $X'$, only deterministic $X'$ yields $H(X')=0$ exactly,
    and hence, by the continuity of the rate functional, this will be
    the best strategy when approaching $q=3/4$ from the left.
  \item $q=0$: From Lemma~\ref{lem:rq0},
    $C_\mathrm{c}|_{q=0}=\max_{y'}\log|\mathcal{S}_{y'}|$. Since
    $|\mathcal{S}_{x'}|=6$ is maximum for all
    $x'\in\mathcal{W}'\triangleq\{$Ser, Leu, Arg$\}$, a distribution
    of $X'$ that maximises~\eqref{eq:rate_q0} is any for which
    $\sum_{x'\in\mathcal{W}'}p(x')=1$. Note that $X'$ needs not be
    deterministic. Capacity is then $C_\mathrm{c}|_{q=0}=\log{6}=
    2.5850$ bits/codon.
  \end{itemize}
  
\noindent\textbf{Remark.}  A trivial upper bound for
any~$ q$ is $C_\mathrm{c}\le C_\mathrm{c}|_{q=0}$. Since
$C_\mathrm{c}|_{q=0}< 3\, C_\mathrm{nc}|_{q=0}=6$, then side-informed
cDNA data embedding capacity will not be able to achieve
non-side-informed ncDNA capacity for every mutation rate. This is
similar to parallel results in side-informed encoding with discrete
hosts~\cite{pradhan03:duality,Barron03} (for uniform side
information), and unlike the well-known result by Costa for continuous
Gaussian hosts~\cite{Costa83}.

From our previous discussion on the value of $C_\mathrm{c}$ for two particular
cases one may conjecture that a pmf with support in $\mathcal{W}'$ may
be capacity-achieving. The actual capacity-achieving strategy is given
by the following theorem:
\begin{thm}\label{thm:capacity}
  Capacity is achieved by the deterministic pmf of $X'$ that maximises
  $H(\mathbf{Z}_{(m)})$.  
\end{thm}

\begin{proof}
   See Appendix~\ref{sec:capac-achi-strat}.
 \end{proof}

\noindent\textbf{Remarks.} Denoting as $\xi'$ the deterministic outcome of
$X'$, it can be numerically verified that $\xi'=\mathrm{Ser}$
maximises~$H(\mathbf{Z}_{(m)})$ for all $\gamma\le 1$, $m$, and $q$,
and thus $C_\mathrm{c}=R_\mathrm{c}^\mathrm{Ser}$ in these
conditions. Some examples of the rates achievable with deterministic
$X'$ are shown in
Figures~\ref{fig:Rc_g1q1e_2}-\ref{fig:Rc_g1e_1q1e_9}. These figures
show that the rates using the linearised approximation given in
Appendix~\ref{sec:capac-achi-strat} are practically indistinguishable
from ones using the Blahut-Arimoto algorithm, whereas the
approximation $p(\mathbf{u}|x')=1/|\mathcal{S}_{x'}|$ is also good but
worsens as $\gamma$ decreases. 

\begin{figure}[t]
     \psfrag{A}[][]{\color{blue}\begin{turn}{30}\tiny Ala\end{turn}}
     \psfrag{R}[][]{\color{blue}\begin{turn}{30}\tiny Arg\end{turn}}
     \psfrag{N}[][]{\color{blue}\begin{turn}{30}\tiny Asn\end{turn}}
     \psfrag{D}[][]{\color{blue}\begin{turn}{30}\tiny Asp\end{turn}}
     \psfrag{C}[][]{\color{blue}\begin{turn}{30}\tiny Cys\end{turn}}
     \psfrag{Q}[][]{\color{blue}\begin{turn}{30}\tiny Gln\end{turn}}
     \psfrag{E}[][]{\color{blue}\begin{turn}{30}\tiny Glu\end{turn}}
     \psfrag{G}[][]{\color{blue}\begin{turn}{30}\tiny Gly\end{turn}}
     \psfrag{H}[][]{\color{blue}\begin{turn}{30}\tiny His\end{turn}}
     \psfrag{I}[][]{\color{blue}\begin{turn}{30}\tiny Ile\end{turn}}
     \psfrag{L}[][]{\color{blue}\begin{turn}{30}\tiny Leu\end{turn}}
     \psfrag{K}[][]{\color{blue}\begin{turn}{30}\tiny Lys\end{turn}}
     \psfrag{M}[][]{\color{blue}\begin{turn}{30}\tiny Met\end{turn}}
     \psfrag{F}[][]{\color{blue}\begin{turn}{30}\tiny Phe\end{turn}}
     \psfrag{P}[][]{\color{blue}\begin{turn}{30}\tiny Pro\end{turn}}
     \psfrag{S}[][]{\color{red}\begin{turn}{30}\tiny Ser\end{turn}}
     \psfrag{T}[][]{\color{blue}\begin{turn}{30}\tiny Thr\end{turn}}
     \psfrag{W}[][]{\color{blue}\begin{turn}{30}\tiny Trp\end{turn}}
     \psfrag{Y}[][]{\color{blue}\begin{turn}{30}\tiny Tyr\end{turn}}
     \psfrag{V}[][]{\color{blue}\begin{turn}{30}\tiny Val\end{turn}}
     \psfrag{B}[][]{\color{blue}\begin{turn}{30}\tiny
         \textit{Stp}\end{turn}}
     \psfrag{m}[t][]{Cascaded mutation stages ($m$)}
     \psfrag{RR}[b][]{$R_\mathrm{c}^{\xi'}$, bits/codon}
     \psfrag{a}[l][l]{\tiny Blahut-Arimoto}
     \psfrag{b}[l][l]{\tiny Linearised approximation}
     \psfrag{c}[l][l]{\tiny $p(\mathbf{u}|x')=1/|\mathcal{S}_{x'}|$}

     \begin{minipage}[b]{0.5\linewidth}
       \includegraphics[width=8cm]{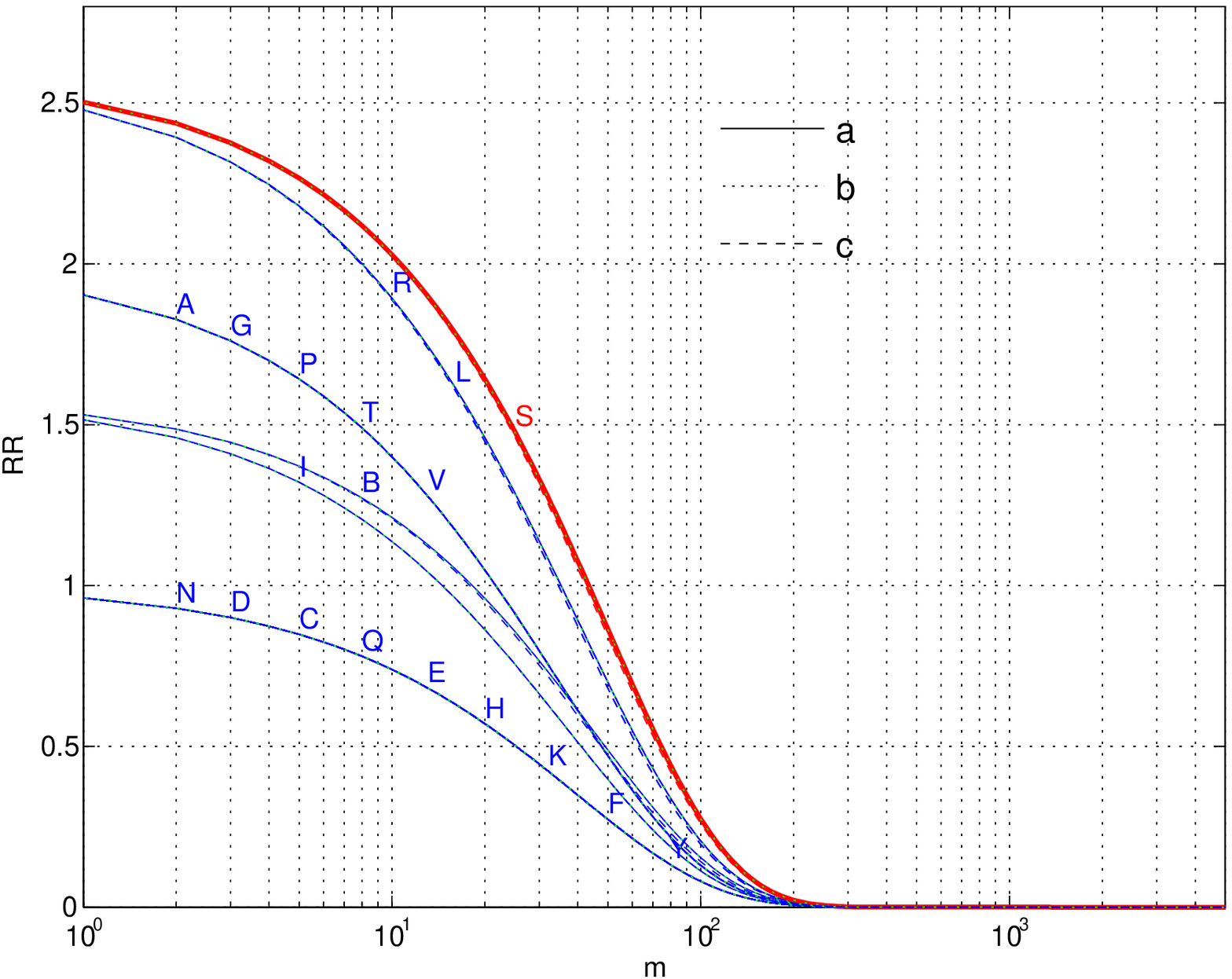}
       \caption{Achievable cDNA data embedding rates for deterministic
         $X'$ ($\gamma=1$, $q=10^{-2}$).}
       \label{fig:Rc_g1q1e_2}
     \end{minipage}~
     \begin{minipage}[b]{0.5\linewidth}
       \includegraphics[width=8cm]{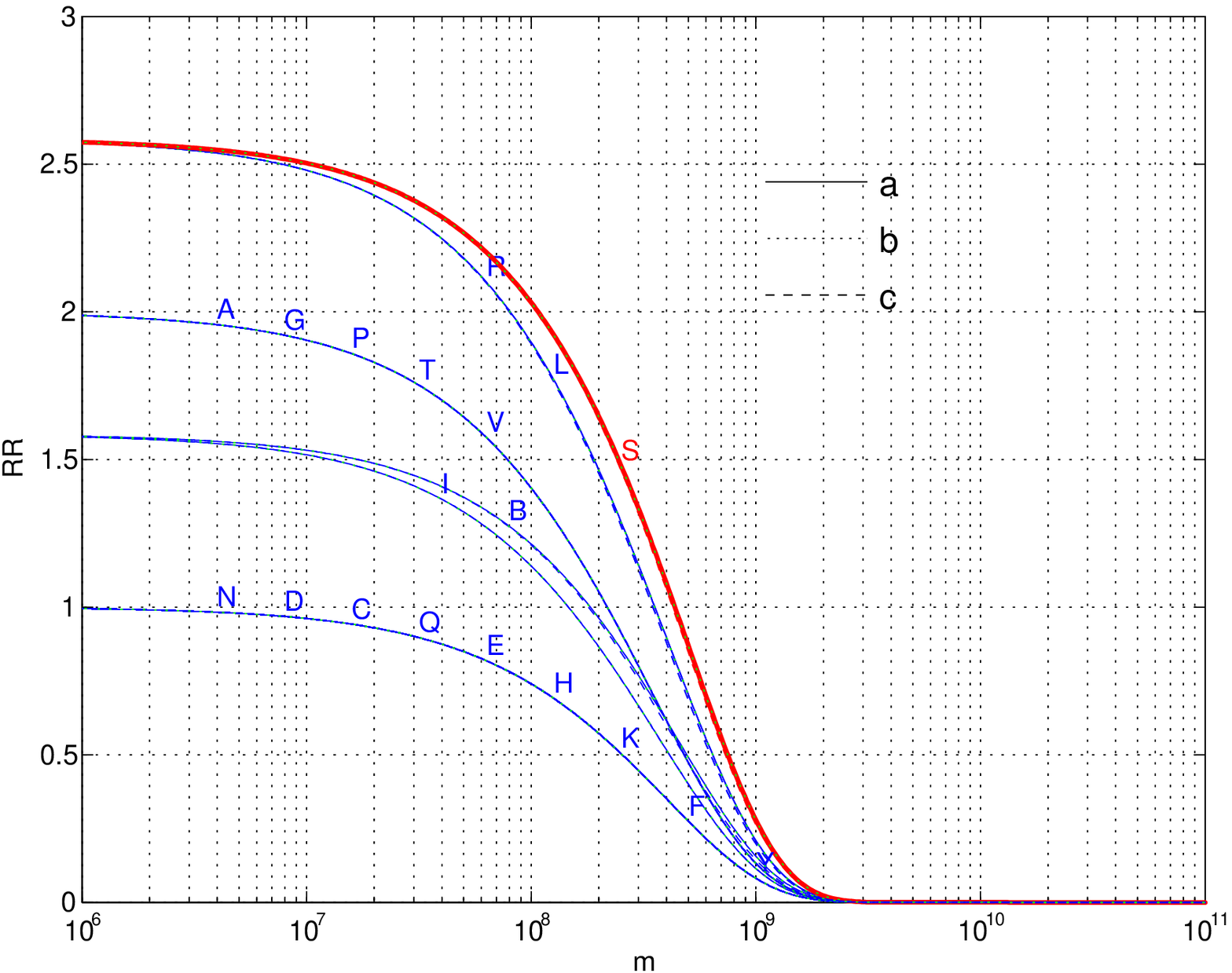}
       \caption{Achievable cDNA data embedding rates for deterministic $X'$ ($\gamma=1$, $q=10^{-9}$).}
       \label{fig:Rc_g1q1e_9}
     \end{minipage}
   \end{figure}

\begin{figure}[t]
     \psfrag{A}[][]{\color{blue}\begin{turn}{30}\tiny Ala\end{turn}}
     \psfrag{R}[][]{\color{blue}\begin{turn}{30}\tiny Arg\end{turn}}
     \psfrag{N}[][]{\color{blue}\begin{turn}{30}\tiny Asn\end{turn}}
     \psfrag{D}[][]{\color{blue}\begin{turn}{30}\tiny Asp\end{turn}}
     \psfrag{C}[][]{\color{blue}\begin{turn}{30}\tiny Cys\end{turn}}
     \psfrag{Q}[][]{\color{blue}\begin{turn}{30}\tiny Gln\end{turn}}
     \psfrag{E}[][]{\color{blue}\begin{turn}{30}\tiny Glu\end{turn}}
     \psfrag{G}[][]{\color{blue}\begin{turn}{30}\tiny Gly\end{turn}}
     \psfrag{H}[][]{\color{blue}\begin{turn}{30}\tiny His\end{turn}}
     \psfrag{I}[][]{\color{blue}\begin{turn}{30}\tiny Ile\end{turn}}
     \psfrag{L}[][]{\color{blue}\begin{turn}{30}\tiny Leu\end{turn}}
     \psfrag{K}[][]{\color{blue}\begin{turn}{30}\tiny Lys\end{turn}}
     \psfrag{M}[][]{\color{blue}\begin{turn}{30}\tiny Met\end{turn}}
     \psfrag{F}[][]{\color{blue}\begin{turn}{30}\tiny Phe\end{turn}}
     \psfrag{P}[][]{\color{blue}\begin{turn}{30}\tiny Pro\end{turn}}
     \psfrag{S}[][]{\color{red}\begin{turn}{30}\tiny Ser\end{turn}}
     \psfrag{T}[][]{\color{blue}\begin{turn}{30}\tiny Thr\end{turn}}
     \psfrag{W}[][]{\color{blue}\begin{turn}{30}\tiny Trp\end{turn}}
     \psfrag{Y}[][]{\color{blue}\begin{turn}{30}\tiny Tyr\end{turn}}
     \psfrag{V}[][]{\color{blue}\begin{turn}{30}\tiny Val\end{turn}}
     \psfrag{B}[][]{\color{blue}\begin{turn}{30}\tiny
         \textit{Stp}\end{turn}}
     \psfrag{m}[t][]{Cascaded mutation stages ($m$)}
     \psfrag{RR}[b][]{\small $R_\mathrm{c}^{\xi'}$, bits/codon}
     \psfrag{a}[l][l]{\tiny Blahut-Arimoto}
     \psfrag{b}[l][l]{\tiny Linearised approximation}
     \psfrag{c}[l][l]{\tiny $p(\mathbf{u}|x')=1/|\mathcal{S}_{x'}|$}

     \begin{minipage}[b]{0.5\linewidth}
       \includegraphics[width=8cm]{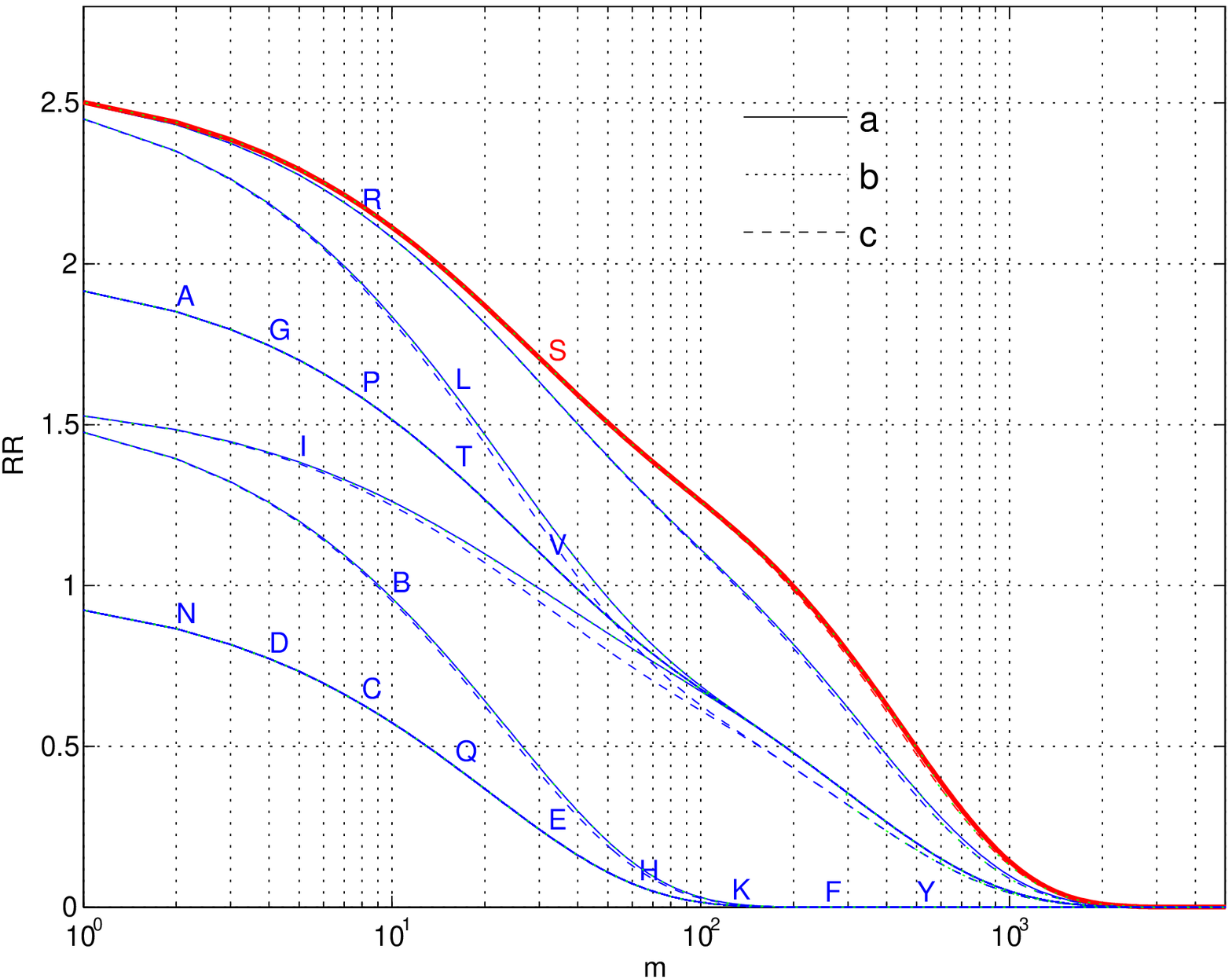}
       \caption{Achievable cDNA data embedding rates for deterministic
         $X'$ ($\gamma=0.1$, $q=10^{-2}$).}
       \label{fig:Rc_g1e_1q1e_2}
     \end{minipage}~
     \begin{minipage}[b]{0.5\linewidth}
       \includegraphics[width=8cm]{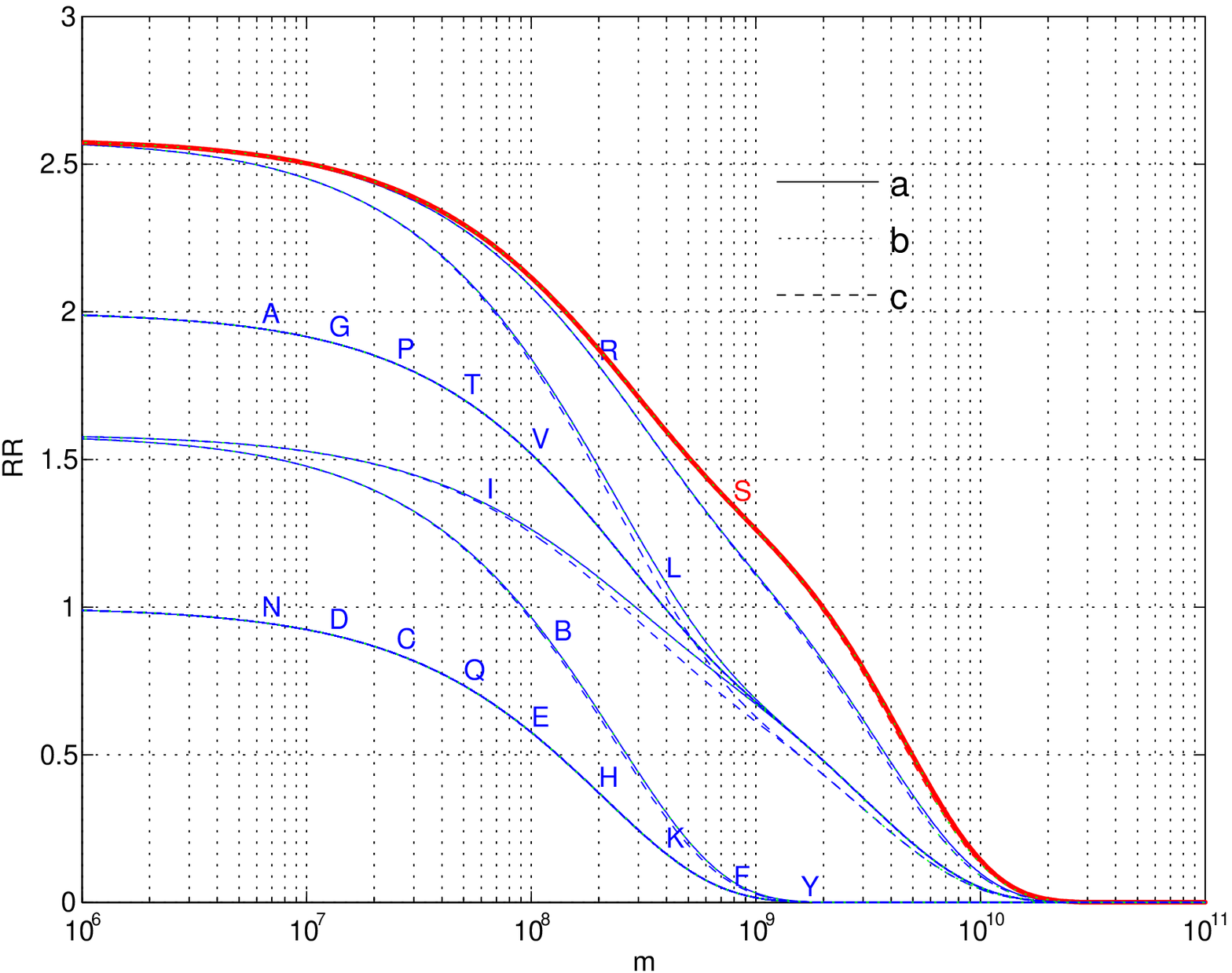}
       \caption{Achievable cDNA data embedding rates for deterministic $X'$ ($\gamma=0.1$, $q=10^{-9}$).}
       \label{fig:Rc_g1e_1q1e_9}
     \end{minipage}
   \end{figure}

\paragraph{Biological interpretations} 
The results in this section solely concern artificial embedding of
information in cDNA, and thus seem to have less obvious applicability
in biological terms than the ones concerning ncDNA. However an
intriguing phenomenon which somehow indicates a biological connection
of these results can be observed in
Figures~\ref{fig:Rc_g1e_1q1e_2}-\ref{fig:Rc_g1e_1q1e_9} which concern
achievable rates for sequences encoding a single amino
acid/symbol. The effect is observed for $\gamma<1$ ---that is, the
range of $\gamma$ in which the model is more realistic--- and consists of
a rate droop for two particular values of $\xi'$ as $m\to\infty$ with
respect to all other symbols presenting the same multiplicity. The
particularity is that these two values of $\xi'$ correspond to the
stop symbol (\textit{Stp}) and the amino acid Leu, which happens to
double as start codon in prokaryotes. Therefore all stop codons and
most of the start codons seem to be less suited to carrying extra
information (redundancy) when isolated. It is not obvious how to
interpret this effect, but one may surmise that these codons might
have suffered some type of selective pressure during the emergence of
the genetic code which somehow depended on the information theoretical
amount studied here. In the case of the stop symbol, this may be due
to the fact that it can only appear once per gene, and so it makes for
a bad conveyor of extra information beyond its basic function.

\section{Conclusions}
\label{sec:conclusions}

We have provided an analysis of the embedding capacity of DNA when
mutations are modelled according to the Kimura model from molecular
evolution studies, and discussed some biological connections of these
results. A more thorough study would require considering insertion and
deletion mutations (\textit{indels}). Although the exact computation
of capacity under indels is an unsolved problem in most digital
communications scenarios, some approximations relying on realignment
methods from bioinformatics might suffice in this
context. Generalisations of the Kimura model may also be considered.
Although in general they will lead to nonsymmetric channels, these can
be numerically handled using the Blahut-Arimoto algorithm.

\appendix

\subsection{Capacity-achieving strategy $p(x')$}
\label{sec:capac-achi-strat}
In order to find the capacity-achieving strategy we need to solve
\begin{equation}
  \frac{\partial}{\partial  p(x')}\left[H(\mathbf{Z}_{(m)})-H(X')+\nu\left(\sum_{y'\in\mathcal{X'}}p(y')-1\right)\right]=0\label{eq:diff}
\end{equation}
for $x'\in\mathcal{X}'$, with $\nu$ a Lagrange multiplier. In the
following we will write
$p(\mathbf{z}|x')=p(\mathbf{Z}_{(m)}=\mathbf{z}|X'=x')$ for notational
convenience.  Assuming natural logarithms for simplicity, and using
$\partial p(\mathbf{z})/\partial p(x')=p(\mathbf{z}|x')$,
\eqref{eq:diff} becomes
\begin{equation}
  \label{eq:difffunctional}
  \sum_{\mathbf{z}\in\mathcal{X}^3} p(\mathbf{z}|x')\log\left( \sum_{y'\in\mathcal{X'}} p(y') p(\mathbf{z}|y')\right)=\log p(x')+\nu,
\end{equation}
for  $x'\in\mathcal{X}'$.  The solution remains unchanged if we
multiply~\eqref{eq:difffunctional} across by $p(x')$. This allows us
to see by inspection that any extreme of the Lagrangian
in~\eqref{eq:diff} has to be deterministic, that is, $p(x')=1$ for
some~$x'=\xi'$ and $p(x')=0$ for $x'\neq \xi'$. Note that this is in
agreement with the strategies for the cases $q=0$ and $\gamma=1$ with
$q=3/4$ discussed in Section~\ref{sec:capacity-analysis}.  See for
instance that a uniform distribution of $X'$ cannot possibly
solve~\eqref{eq:difffunctional} for all $x'\in\mathcal{X}'$, because
$\sum_{\mathbf{z}} p(\mathbf{z}|x')\log\left(\sum_{y'}
  p(\mathbf{z}|y')\right)$ is not constant on $x'$ unless $\gamma=1$
and $q= 3/4$, in which case we have shown that capacity is zero for
any distribution.

According to the previous discussion, for any capacity-achieving
solution it always holds that $H(X')=0$, and then we just have to
maximise $H(\mathbf{Z}_{(m)})$ over the ensemble of 21 deterministic
distributions of $X'$. 

The computation of $R_\mathrm{c}^{\xi'}$ and of the maximising
distribution $\mathbf{U}|\xi'$ can be done using the Blahut-Arimoto
algorithm, following the discussion in
Section~\ref{sec:achi-rate-analys} on the optimal strategy for fixed
$p(x')$. Note that $\xi'=\mathrm{Trp}$ and $\xi'=\mathrm{Met}$ can be
ruled out outright, since
$|\mathcal{S}_\mathrm{Trp}|=|\mathcal{S}_\mathrm{Met}|=1$, and then
only null rates are possible in these cases. Then we only need to
compute $R_\mathrm{c}^{\xi'}$ for 19 amino acids. Also,
$\xi'=\mathit{Stp}$ can only be considered hypothetically, since this
symbol can only appear exactly once in a gene.

\paragraph{Approximation to maximising strategy} It is also possible
to provide a closed-form approximation to the maximising distribution
$\mathbf{U}|\xi'$, which yields a better approximation to the
embedding rate than just using the approximation
$p(\mathbf{u}|\xi')=1/|\mathcal{S}_{\xi'}|$ discussed in
Section~\ref{sec:achi-rate-analys}. Observe firstly that when $X'$ is
deterministic the situation is equivalent to a non-side informed
discrete channel with $|\mathcal{S}_{\xi'}|$ inputs and
$|\mathcal{X}|^3$ outputs, with a transition probability matrix
${\boldsymbol\Lambda}$ whose rows are the rows of ${\boldsymbol
  \Pi}^m$ corresponding to the codons associated with~$\xi'$. In
general this channel will not be symmetric nor weakly symmetric, since
although its rows are permutations of the same set of probabilities,
its columns are not, and their sum is not constant either. However
$H(\mathbf{Z}_{(m)}|\mathbf{U})$ is still independent of the
distribution of $\mathbf{U}$, and then we only need to maximise
$H(\mathbf{Z}_{(m)})$ to find capacity. The corresponding conditions for
the maximum are
\begin{equation}
  \label{eq:cond_approx}
  \sum_{\mathbf{z}\in\mathcal{X}^3} p(\mathbf{z}|\mathbf{v})\log p(\mathbf{z})+1=\rho,
\end{equation}
for $\mathbf{v}\in\mathcal{S}_{\xi'}$, and with $\rho$ a Lagrange
multiplier.

Using $\log x\le x-1$ and
$p(\mathbf{z})=\sum_{\mathbf{u}\in\mathcal{S}_{\xi'}}
p(\mathbf{z}|\mathbf{u})p(\mathbf{u}|\xi')$,
we can  write
\begin{equation}
  \label{eq:cond_approx_lin}
  \sum_{\mathbf{z}\in\mathcal{X}^3} p(\mathbf{z}|\mathbf{v})\sum_{\mathbf{u}\in\mathcal{S}_{\xi'}} p(\mathbf{z}|\mathbf{u})p(\mathbf{u}|\xi')\le\rho,
\end{equation}
for $\mathbf{v}\in\mathcal{S}_{\xi'}$. Our approximation consists of
solving $p(\mathbf{u}|\xi')$ by enforcing equality
in~\eqref{eq:cond_approx_lin} for all
$\mathbf{v}\in\mathcal{S}_{\xi'}$. This yields the linear system
\begin{equation}
  \label{eq:linearised}
  {\boldsymbol \pi} \left({\boldsymbol\Lambda}{\boldsymbol\Lambda}^T\right) = \rho \mathbf{1},
\end{equation}
where the probabilities $p(\mathbf{u}|\xi')$, with
$\mathbf{u}\in\mathcal{S}_{\xi'}$, are the elements of the
$1\times|\mathcal{S}_{\xi'}|$ vector ${\boldsymbol \pi}$ (arranged in
the same codon order as the rows of ${\boldsymbol\Lambda}$), and
$\mathbf{1}$ is an all-ones vector of size
$1\times|\mathcal{S}_{\xi'}|$. Since ${\boldsymbol \pi}$ must be a
pmf, we may fix any arbitrary value of $\rho$, such as $\rho=1$, and
then normalise the solution $\widetilde{\boldsymbol\pi}$ to the
resulting linear system, that is
\begin{equation}\label{eq:widetildepi}
  \widetilde{\boldsymbol\pi}=\mathbf{1}({\boldsymbol\Lambda}
  {\boldsymbol\Lambda}^T)^{-1}.
\end{equation}
The matrix ${\boldsymbol\Lambda} {\boldsymbol\Lambda}^T$~is invertible
if both $q\neq 1/(4\gamma/3)$ and $q\neq 1/(2(1-\gamma/3))$ because in
this case the rows of ${\boldsymbol\Lambda}$ are linearly
independent. This is due to the fact that under the two conditions
above the rows of ${\boldsymbol\Pi}^m$ are linearly independent, since
its eigenvalues are all the possible products of three eigenvalues of
$\Pi^m$~\cite{magnus99:matrix} and the conditions above
guarantee that these are nonzero. A sufficient condition for
the invertibility of ${\boldsymbol\Lambda} {\boldsymbol\Lambda}^T$ is
$q<1/2$, which spans most cases of interest.

Since we have linearised the optimisation problem then
$\widetilde{\boldsymbol\pi}$ may contain negative values, but in
practice these are relatively small. Setting these values to zero and
normalising $\widetilde{\boldsymbol\pi}$ we obtain an approximation to
the optimum distribution $p(\mathbf{u}|\xi')$. An example of this
approximation compared to the results of the Blahut-Arimoto algorithm
is shown in Figure~\ref{fig:approx}.
\begin{figure}[t]
  \centering
  \psfrag{a}[l][l]{\scriptsize Blahut-Arimoto}
  \psfrag{b}[l][l]{\scriptsize Linearised approximation}
  \psfrag{CTA}[r][r]{\tiny CTA}
  \psfrag{CTC}[r][r]{\tiny CTC}
  \psfrag{CTT}[r][r]{\tiny CTT}
  \psfrag{CTG}[r][r]{\tiny CTG}
  \psfrag{TTA}[r][r]{\tiny TTA}
  \psfrag{TTG}[r][r]{\tiny TTG}
  \psfrag{puxp}[b][]{$p(\mathbf{u}|\xi')$}
  \includegraphics[width=8cm]{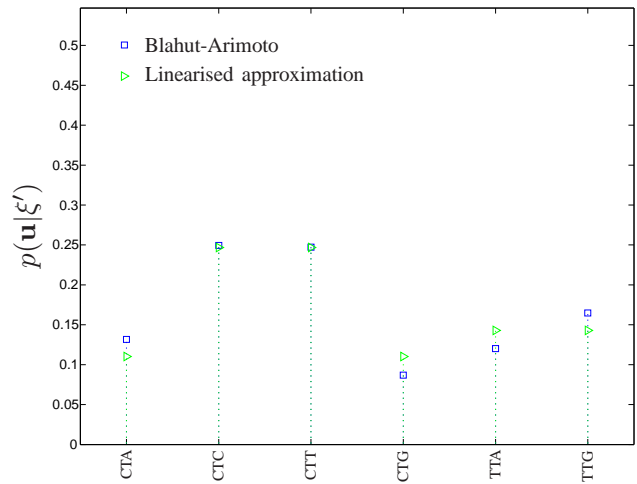}
  \caption{Comparison of maximising $p(\mathbf{u}|\xi')$ distributions
    for the deterministic case $\xi'=\mathrm{Leu}$ ($q=10^{-2}$,
    $m=100$, $\gamma=0.1$)}
  \label{fig:approx}
\end{figure}

\bibliographystyle{IEEEtran}   
\bibliography{it}

\begin{thebibliography}{10}
\providecommand{\url}[1]{#1}
\csname url@samestyle\endcsname
\providecommand{\newblock}{\relax}
\providecommand{\bibinfo}[2]{#2}
\providecommand{\BIBentrySTDinterwordspacing}{\spaceskip=0pt\relax}
\providecommand{\BIBentryALTinterwordstretchfactor}{4}
\providecommand{\BIBentryALTinterwordspacing}{\spaceskip=\fontdimen2\font plus
\BIBentryALTinterwordstretchfactor\fontdimen3\font minus
  \fontdimen4\font\relax}
\providecommand{\BIBforeignlanguage}[2]{{%
\expandafter\ifx\csname l@#1\endcsname\relax
\typeout{** WARNING: IEEEtran.bst: No hyphenation pattern has been}%
\typeout{** loaded for the language `#1'. Using the pattern for}%
\typeout{** the default language instead.}%
\else
\language=\csname l@#1\endcsname
\fi
#2}}
\providecommand{\BIBdecl}{\relax}
\BIBdecl

\bibitem{clelland99a}
C.~T. Clelland, V.~Risca, and C.~Bancroft, ``Hiding messages in {DNA}
  microdots,'' \emph{Nature}, vol. 399, no. 6736, pp. 533--534, June 1999.

\bibitem{cox01:long_term}
J.~P. Cox, ``Long-term data storage in {DNA},'' \emph{Trends in Biotechnology},
  vol.~19, no.~7, pp. 247--250, July 2001.

\bibitem{Shimanovsky02}
B.~Shimanovsky, J.~Feng, and M.~Potkonjak, ``Hiding data in {DNA},'' in
  \emph{Procs. of the 5th Intl. Workshop in Information Hiding},
  Noordwijkerhout, The Netherlands, October 2002, pp. 373--386.

\bibitem{wong03:organ_data}
P.~C. Wong, K.~Wong, and H.~Foote, ``Organic data memory using the {DNA}
  approach,'' \emph{Comms. of the ACM}, vol.~46, no.~1, pp. 95--98, January
  2003.

\bibitem{arita04:secret}
M.~Arita and Y.~Ohashi, ``Secret signatures inside genomic {DNA},''
  \emph{Biotechnol. Prog.}, vol.~20, no.~5, pp. 1605--1607, September-October
  2004.

\bibitem{modegi05:water_embed}
T.~Modegi, ``Watermark embedding techniques for {DNA} sequences using codon
  usage bias features,'' in \emph{16th Intl. Conf. on Genome Informatics},
  Yokohama, Japan, December 2005.

\bibitem{yachie07:alignnment}
N.~Yachie, K.~Sekiyama, J.~Sugahara, Y.~Ohashi, and M.~Tomita,
  ``Alignment-based approach for durable data storage into living organisms,''
  \emph{Biotechnol. Prog.}, vol.~23, no.~2, pp. 501--505, April 2007.

\bibitem{heider07:dna_based_watermarking}
D.~Heider and A.~Barnekow, ``{DNA}-based watermarks using the {DNA}-{C}rypt
  algorithm,'' \emph{BMC Bioinformatics}, vol.~8, no. 176, February 2007.

\bibitem{heider09:dna_noncoding}
D.~Heider, M.~Pyka, and A.~Barnekow, ``{DNA} watermarks in non-coding
  regulatory sequences,'' \emph{BMC Research Notes}, vol.~2, no. 125, July
  2009.

\bibitem{daniel08:complete}
D.~Gibson, G.~Benders, C.~Andrews-Pfannkoch, E.~Denisova, H.~Baden-Tillson,
  J.~Zaveri, T.~Stockwell, A.~Brownley, M.~A. D.~W.~Thomas, C.~Merryman,
  L.~Young, V.~Noskov, J.~Glass, J.~Venter, C.~Hutchison, and H.~Smith,
  ``Complete chemical synthesis, assembly, and cloning of a mycoplasma
  genitalium genome,'' \emph{Science}, vol. 319, pp. 1215--1219, 2008.

\bibitem{shannon48:math}
C.~E. Shannon, ``A mathematical theory of communication,'' \emph{Bell System
  Technical Journal}, vol.~27, pp. 379--423 and 623--656, July and October
  1948.

\bibitem{kimura80}
M.~Kimura, ``A simple method for estimating evolutionary rate in a finite
  population due to mutational production of neutral and nearly neutral base
  substitution through comparative studies of nucleotide sequences,'' \emph{J.
  Molec. Biol.}, vol.~16, pp. 111--120, 1980.

\bibitem{purvis97:ti_tv}
A.~Purvis and L.~Bromham, ``Estimating the transition/transversion ratio from
  independent pairwise comparisons with an assumed phylogeny,'' \emph{Journal
  of Molecular Evolution}, vol.~44, pp. 112--119, 1997.

\bibitem{kunkel04:dna_replication}
T.~A. Kunkel, ``{DNA} replication fidelity,'' \emph{J. Biol. Chem.}, vol. 279,
  no.~17, pp. 16\,895--16\,898, April 2004.

\bibitem{li97:_molecular_evolution}
W.~Li, \emph{Molecular Evolution}.\hskip 1em plus 0.5em minus 0.4em\relax
  Sinauer Associates, 1997.

\bibitem{magnus99:matrix}
J.~R. Magnus and H.~Neudecker, \emph{Matrix Differential Calculus with
  Applications in Statistics and Econometrics}, 3rd~ed.\hskip 1em plus 0.5em
  minus 0.4em\relax John Wiley \& Sons, 1999.

\bibitem{may03:_detection}
E.~May, M.~Rintoul, A.~Johnston, W.~Hart, J.~Watson, and R.~Pryor, ``Detection
  and reconstruction of error control codes for engineered and biological
  regulatory systems,'' Sandia National Laboratories, Tech. Rep., 2003.

\bibitem{gutfraind06:error_tolerant}
A.~Gutfraind, ``Error-tolerant coding and the genetic code,'' Master's thesis,
  University of Waterloo, 2006.

\bibitem{battail07:infotheory}
G.~Battail, ``Information theory and error-correcting codes in genetics and
  biological evolution,'' in \emph{Introduction to Biosemiotics}, M.~Barbieri,
  Ed.\hskip 1em plus 0.5em minus 0.4em\relax Springer, 2007.

\bibitem{may07:_bits_bases}
E.~May, ``Bits and bases: An analysis of genetic information paradigms,'' in
  \emph{41st Asilomar Conference on Signals, Systems and Computers (ACSSC)},
  Asilomar, USA, November 2007, pp. 165--169.

\bibitem{Gilbert1986}
W.~Gilbert, ``Origin of life: The rna world,'' \emph{Nature}, vol. 319, no.
  6055, pp. 618--618, Feb 1986.

\bibitem{fu01:estimating}
Y.~Fu, ``Estimating mutation rate and generation time from longitudinal samples
  of {DNA} sequences,'' \emph{Mol. Biol. and Evolution}, vol.~18, no.~4, pp.
  620--626, 2001.

\bibitem{pradhan03:duality}
S.~S. Pradhan, J.~Chou, and K.~Ramchandran, ``Duality between source coding and
  channel coding and its extension to the side information case,'' \emph{IEEE
  Trans. on Inf. Theory}, vol.~49, no.~5, pp. 1181--1203, May 2003.

\bibitem{Barron03}
R.~J. Barron, B.~Chen, and G.~W. Wornell, ``The duality between information
  embedding and source coding with side information and some applications,''
  \emph{IEEE Trans. on Inf. Theory}, vol.~49, no.~5, pp. 1159--1180, May 2003.

\bibitem{Gelfand}
S.~I. Gel'fand and M.~S. Pinsker, ``Coding for channel with random
  parameters,'' \emph{Problems of Control and Information Theory}, vol.~9,
  no.~1, pp. 19--31, 1980.

\bibitem{blahut72:_comput}
R.~Blahut, ``Computation of channel capacity and rate-distortion functions,''
  \emph{Information Theory, IEEE Transactions on}, vol.~18, no.~4, pp. 460 --
  473, Jul. 1972.

\bibitem{dupuis04:_blahut}
F.~Dupuis, W.~Yu, and F.~Willems, ``{B}lahut-{A}rimoto algorithms for computing
  channel capacity and rate-distortion with side information,'' in \emph{Intl.
  Symposium on Information Theory (ISIT)}, June-July 2004, p. 179.

\bibitem{Costa83}
M.~H. Costa, ``Writing on dirty paper,'' \emph{IEEE Trans. on Information
  Theory}, vol.~29, no.~3, pp. 439--441, May 1983.

\end{thebibliography}

\end{document}